\documentclass[twocolumn,natbib,envcountsame,numbook,dvipsnames,pagebackref]{svjour3}
\smartqed
\pdfoutput=1
\usepackage[utf8]{inputenc}
\usepackage[T1]{fontenc}
\usepackage{amsmath,amssymb}
\usepackage{mathtools}
\usepackage{txfonts}
\usepackage{microtype}
\usepackage{tikz}
\usepackage{doi}
\usepackage{subcaption}
\usepackage{paralist}
\usepackage[strings]{underscore}

\usepackage{hyphenat}
\usepackage{dsfont}
\usepackage{flushend}
\usepackage[sort&compress,noabbrev,capitalize]{cleveref}

\renewcommand*{\backref}[1]{}
\renewcommand*{\backrefalt}[4]{({%
    \ifcase #1 Not cited.%
          \or Cited on page~#2%
          \else Cited on pages #2%
    \fi%
    })}

\newcommand{\poly}{\text{poly}}
\newcommand{\fp}{fixed\hyp parameter}

\newcommand{\pt}{polynomial\hyp time}
\newcommand{\IS}{\textsc{Maximum Independent Set}}
\newcommand{\MIS}{\textsc{Max-Weight Independent Set}}
\newcommand{\MCIS}{\textsc{Max-Weight Colorful Independent Set}}
\newcommand{\cCS}{\textsc{Maximum \(c\)\hyp Colorable Subgraph}}
\newcommand{\McCS}{\textsc{Max-Weight \(c\)\hyp Colorable Subgraph}}
\newcommand{\MCC}{\textsc{Multicolored Clique}}

\newcommand{\HC}{\textsc{Hamiltonian Cycle}}

\newcommand{\vrt}[2]{\ensuremath{v_{\smash{#1}}^{\smash{#2}}}}
\newcommand{\Vrt}[1]{\ensuremath{V_{\smash{#1}}}}
\newcommand{\urt}[2]{\ensuremath{u_{\smash{#1}}^{\smash{#2}}}}
\newcommand{\Urt}[1]{\ensuremath{U_{\smash{#1}}}}

\newcommand{\ert}[2]{\ensuremath{e_{\smash{#1}}^{\smash{#2}}}}
\newcommand{\Ert}[1]{\ensuremath{E_{\smash{#1}}}}

\usetikzlibrary{arrows,decorations.pathmorphing,backgrounds,positioning,fit,shapes,calc}
\hypersetup{allcolors=blue,colorlinks} %

\spnewtheorem{construction}[theorem]{Construction}{\bfseries}{\normalfont}

\crefname{construction}{Construction}{Constructions}
\crefname{figure}{Figure}{Figures}
\DeclareMathOperator{\col}{col}

\newcommand{\ind}[1]{$#1$\hyp independent}
\newcommand{\indc}[1]{$#1$\hyp independence}
\newcommand{\problemdef}[3]{%
  \begin{problem}[#1]
  \begin{compactdesc}
  \item[\it Input:] #2
  \item[\it Task:] #3
  \end{compactdesc}
\end{problem}
\par
}

\newcommand{\decproblemdef}[3]{%
  \begin{problem}[#1]
  \begin{compactdesc}
  \item[\it Input:] #2
  \item[\it Question:] #3
  \end{compactdesc}
\end{problem}
\par
}

\title{\boldmath Inductive \(k\)-independent graphs and $c$\hyp colorable subgraphs in scheduling: A review}
\titlerunning{Inductive \(k\)-independent graphs and $c$\hyp colorable subgraphs in scheduling: A review}

\author{Matthias Bentert
  \and
  René van Bevern\thanks{René van Bevern was supported
    by grant 16-31-60007 mol\textunderscore{}a\textunderscore{}dk
    of the Russian Foundation for Basic Research.
    This work was initiated during a research visit of Ren\'{e} van Bevern
    to TU Berlin in June~2017, partly supported by TU~Berlin
    and by the Ministry of Science and Education
    of the Russian Federation under the 5-100 Excellence Programme.
}
  \and
  Rolf~Niedermeier}

\institute{
  Matthias Bentert \and Rolf Niedermeier \at
  Institut f\"ur Softwaretechnik und Theoretische Informatik, TU~Berlin, Berlin, Germany, \email{matthias.bentert@tu-berlin.de, rolf.niedermeier@tu-berlin.de}
  \and
  René van Bevern \at
  Department of Mechanics and Mathematics, Novosibirsk State University, Novosibirsk, Russian Federation, \email{rvb@nsu.ru}
  \\[.5em]
  Sobolev Institute of Mathematics of the Siberian Branch of the Russian Academy of Sciences, Novosibirsk, Russian Federation}

\begin{document}

\maketitle

\abstract{
  Inductive \ind{k} graphs
  generalize chordal graphs
  and
  have recently been advocated %
  in the context of interference\hyp avoiding
  wireless communication scheduling.
  The NP-hard problem
  of
  finding maximum\hyp weight induced $c$\hyp colorable subgraphs,
  which is a %
  generalization 
  of finding maximum independent sets,
  naturally occurs when selecting \(c\)~sets
  of pairwise non\hyp conflicting jobs
  (modeled as graph vertices).
  We investigate the parameterized complexity
  of this problem
  on inductive \ind{k} graphs.
  We show that the \IS{} problem is W[1]-hard
  even on 2-simplicial 3-minoes---a subclass of inductive 2\hyp independent graphs.
  In contrast,
  we prove that the more general \McCS{} problem
  is fixed\hyp parameter tractable
  on edge-wise unions of cluster and chordal graphs,
  which are 2-simplicial.
  In both cases, 
  the parameter is the solution size.
  Aside from this, %
  we survey %
  other graph classes %
  between inductive \ind{1}
  and inductive \ind{2} graphs
  with applications in scheduling.
}

\keywords{
  independent set\and
  job interval selection\and
  interval graphs\and
  chordal graphs\and
  inductive \ind{k} graphs\and
  NP-hard problems\and
  parameterized complexity
}

\section{Introduction}
Finding sets of ``independent'' (that is,
pairwise non\hyp conflicting)
jobs is of central importance in many
scheduling scenarios. %
In particular,
the NP-hard \IS{} problem
and its generalization
of finding maximum\hyp weight induced \(c\)\hyp colorable subgraphs,
both in
inductive \ind{k} graphs,
have recently been identified as key tools in the development of
polynomial\hyp time approximation algorithms in
interference\hyp avoiding
wireless communication scheduling
\citep{AHT17,Hal16,HT15}.

We conduct %
a deeper study
of the computational complexity
of these problems
on subclasses of inductive \ind{k} 
graphs. %
We also exhibit a rich fine structure
of known graph classes
that have applications in scheduling
and are inductive \ind{2}.

A graph is inductive \ind{k}
if its vertices can be ordered from left to right
so that
the ``right-hand'' neighborhood
of each vertex
contains no independent vertex set of 
size greater than~$k$.
The %
problem we study on these graphs  %
is \textsc{Maximum (Weight) \(c\)\hyp Colorable Subgraph}\footnote{\cite{AHT17}
  are interested
  in maximum\hyp weight
  unions of $c$~independent sets.
  In the graph theory literature the problem is known as
\McCS{}; we prefer to stick to the 
established graph theory notion.}---find a maximum 
(weight) vertex  subset that induces a subgraph allowing
for coloring the vertices in \(c\)~colors so that
no adjacent vertices have the same color.
Clearly, for $c=1$,
we obtain
the classical \textsc{Maximum (Weight) Independent Set} problem.

\begin{figure*}
\begin{center}
\begin{tikzpicture}
\node[draw, fill=white] (I3I) at (7,6) {inductive \ind{3}};

\node[draw, fill=white] (N4I) at (10,5) {$K_{1,5}$-free};

\node[draw, fill=white] (I2I) at (4,5) {\textbf{\boldmath inductive \ind{2}}} edge (I3I);

\node[draw, fill=white] (2S) at (4,4) {\textbf{2\hyp simplicial}} edge (I2I);

\node[draw, fill=white] (N3I) at (7,4) {\textbf{\boldmath$K_{1,4}$-free}} edge (I3I) edge (N4I);
\node[draw, fill=white, align=center] (UIUI) at (10,3.8)
{
  unit interval $\bowtie$ unit interval
  \\
  \cite{Jia10}
} edge (I3I) edge (N4I);

\node[draw, fill=white, align=center] (CC) at (2,2.5) {
  \textbf{\boldmath cluster $\bowtie$ chordal}
  \\
  \cref{thm:iseasy}
} edge (2S);
\node[draw, fill=white, align=center] (SM) at (5.5,2.5) {\textbf{\boldmath2-simplicial $\cap$ 3-mino}\\ \cref{thm:ishard}} edge (2S) edge (N3I);

\node[draw, fill=white, align=center] (strip) at (0,1) {cluster $\bowtie$ interval\\\cref{sec:cluster-interval}} edge (CC);
\node[draw, fill=white, align=center] (chordal) at (4,1) {inductive \ind{1}\\\cref{sec:chordal}} edge (CC);
\node[draw, fill=white, align=center] (N2I) at (10,1) {$K_{1,3}$-free\\\cref{sec:k13free}} edge (N3I) edge (I2I);

\node[draw, fill=white, align=center] (interval) at (2,-0.2) {interval\\\cref{sec:interval}} edge (chordal) edge (strip);

\begin{pgfonlayer}{background}
	\draw[fill=white, rounded corners]
   		(0.6,-0.8) --
		(3.4,-0.8) -- 
		(3.4,0.4) -- 
		(0.6,0.4) -- cycle;
		
	\draw[fill=black!20!white, rounded corners]
   		(2,0.4) --
		(11.3,0.4) -- 
		(11.3,1.65) -- 
		(2,1.65) -- cycle;
		
	\draw[fill=black!40!white, rounded corners]
   		(-1.6,0.4) --
		(2,0.4) -- 
		(2,1.65) -- 
		(3.7,1.65) -- 
		(3.7,3.25) -- 
		(-1.6,3.25) -- cycle;

	\draw[fill=black!60!white, rounded corners]
   		(2,3.25) --
		(3.7,3.25) -- 
		(3.7,1.65) -- 
		(12.1,1.65) -- 
		(12.1,6.6) -- 
		(2,6.6) -- cycle;
\end{pgfonlayer}
		
\end{tikzpicture}
\caption{For definitions of graph classes, see \cref{sec:prel}.  Our new results are in bold.
  \newline
  White --- \McCS{} is solvable in polynomial time.
  \newline
  Light Gray --- \MIS{} is \pt{} solvable, \cCS{} is NP-hard, \McCS{} parameterized by the solution size is \fp{} tractable.
  \newline
  Gray --- \IS{} is NP-hard. \McCS{} parameterized by solution size is \fp{} tractable.
  \newline
  Dark Gray --- \IS{} is NP-hard and W[1]-hard parameterized by the solution size.}
\label{fig:MIS}
\end{center}
\end{figure*}

\cite{HT15} found that,
in interference\hyp{}avoiding wireless communication
scheduling,
one typically faces
these problems
in inductive \ind{k} graphs with $k\leq12$.
Unfortunately,
we will see that already unweighted \IS{}
is hard for subclasses of inductive \ind{2} graphs.
Notably, however,
classic scheduling problems 
have been studied for a number of subclasses of inductive \ind{2} graphs,
including interval graphs \citep{KLPS07}, strip graphs \citep{HK06},
their superclass 2-track interval graphs \citep{HKML11,BMNW15},
and claw-free graphs \citep{GK03,KEGGM17,KM17}.

\paragraph{Our contributions and organization of this work.}
Our main contributions are as follows.
We refer to \cref{sec:prel} for definitions
of graph and complexity classes.

In \cref{sec:classes},
we explore the fine structure of 
graph classes
below 
inductive \ind{3} graphs,
discuss relations to known graph classes
(often already appearing in 
scheduling applications), and also survey the recognition 
complexities of the respective graph classes.

In \cref{sec:ishard},
we show that %
already (unweighted)
\IS{}
parameterized by the solution size (number of vertices)
is W[1]-hard on 2\hyp simplicial 3-minoes, a proper subclass 
of inductive \ind{2} graphs. 
Thus,
there is no reason
to hope for exact fixed\hyp parameter algorithms
to find even small independent sets
in these graphs.

In contrast,	
in \cref{sec:iseasy}
we show that
\McCS{} parameterized by the solution size
is fixed\hyp parameter tractable on
a class that lies
properly between inductive \ind{1} and inductive \ind{2} graphs
and generalizes the class of strip graphs,
which \citet{HK06}
used to model
the \textsc{Job Interval Selection} problem.

Finally,
in \cref{sec:related},
we briefly survey complexity results
on \McCS{} in the other graph classes
discussed in \cref{sec:classes}.

We refer to \cref{fig:MIS} for an overview of results
on the parameterized complexity %
of \IS{} and \McCS{}
on graph
classes mainly below inductive \ind{3} graphs.

\paragraph{Related work.}
It is well known that, on general graphs,
already (unweighted) \IS{}
is
a notoriously hard problem
from the viewpoint of polynomial-time approximability
as well as from the viewpoint of parameterized complexity.
Hence, it is natural to study \IS{} and related problems on special 
graph classes.

Our results improve on or complement
the following known %
results on \IS{} and \McCS{} on
inductive \ind{k} graph classes.
\begin{itemize}
\item \IS{} parameterized 
by solution size is fixed\hyp parameter tractable on strip 
graphs \citep{BMNW15}; strip graphs are equivalent to
the class of cluster $\bowtie$ interval graphs (see \cref{fig:MIS}),
a proper subclass of inductive \ind{2} graphs.
\item \MIS{} parameterized
by solution size is W[1]-hard already on unit 2-track interval 
graphs \citep{Jia10}; unit 2-track interval graphs
form a subclass of inductive \ind{3} graphs \citep{YB12}.
\item (Unweighted) \cCS{} parameterized 
by solution size is fixed\hyp parameter tractable on inductive \ind{1}
(that is, chordal) graphs \citep{MPR+13}.
\end{itemize}

\section{Preliminaries}
\label{sec:prel}
This section introduces
basic notation and concepts that we will use throughout this work.

\paragraph{Fixed\hyp parameter algorithms.}
The essential idea behind \fp{} algorithms is
to accept exponential running times,
which are seemingly inevitable in solving NP-hard problems,
but to restrict them to one aspect of the problem,
the \emph{parameter} \citep{CFK+15,DF13,FG06,Nie06}.

Thus, formally, an instance of a \emph{parameterized problem}~\(\Pi\subseteq\Sigma^*\times\mathds{N}\) is a pair~$(x,k)$ consisting of the \emph{input~$x$} and the \emph{parameter~$k$}.
A parameterized problem~$\Pi$ is \emph{\fp{} tractable~(FPT)} with respect to a parameter~$k$ if there is an algorithm solving any instance of~$\Pi$ with size~$n$ in $f(k) \cdot \poly(n)$~time for some computable function~$f$.  Such an algorithm is called a \emph{\fp{} algorithm}. It is potentially efficient for small values of~$k$, in contrast to an algorithm that is merely running in polynomial time for each fixed~$k$ (thus allowing the degree of the polynomial to depend on~$k$).  FPT is the complexity class of \fp{} tractable parameterized problems.  

\paragraph{Parameterized intractability} To show that a problem is presumably not \fp{} tractable, there is a parameterized analog of NP\hyp hardness theory.
The parameterized analog of P $\subseteq$ NP is
a hierarchy of complexity classes
FPT\({}\subseteq{}\)W[1]\({}\subseteq{}\)W[2]\({}\subseteq{}\ldots\),
where all inclusions are assumed to be proper.
Analogously to the class coNP,
for \(t\in\mathds{N}\),
one defines the classes
coW[$t$]
containing all parameterized problems~$\Pi\subseteq\Sigma^*\times\mathds{N}$
for which the complement~$(\Sigma^*\times\mathds{N})\setminus\Pi$ is in W[$t$].
A parameterized problem~$\Pi$ with parameter~\(k\) is
\emph{(co)W[$t$]-hard} for some \(t\in\mathds{N}\)
if any problem in (co)W[$t$] has a parameterized reduction to~\(\Pi\):
a \emph{parameterized reduction} from a parameterized problem~$\Pi_1$
to a parameterized problem~$\Pi_2$ is an algorithm mapping an instance~$(x,k)$
to an instance~$(x',k')$
in~$f(k)\cdot\poly(|x|)$~time
such that $k'\leq g(k)$ and
\((x,k)\in\Pi_1\iff(x',k')\in\Pi_2\),
where \(f\)~and~\(g\) are arbitrary computable functions.
By definition,
no (co)W[$t$]-hard problem is \fp{} tractable
unless FPT${}={}$W[$t$].

\paragraph{Graph theory.}
We consider undirected,
finite,
simple graphs
$G=(V,E)$,
where $V$~is the set of \emph{vertices}
and $E \subseteq \{\{v,w\}\mid v\ne w\text{ and }v,w\in V\}$~is the set of \emph{edges}.

If the vertices have \emph{weights}~$w \colon V \to \mathds{N}$, %
then
we denote the weight of a vertex subset~$S\subseteq V$
as~$w(S):=\sum_{v\in S}w(v)$.

The \emph{(open) neighborhood} \(N_G(v):=\{u\in V\mid \{u,v\}\in E\}\) of a vertex~\(v\in V\)
is the set of vertices \emph{adjacent to~\(v\)} in~\(G\).
The \emph{closed neighborhood} of~\(v\)
is \(N_G[v]:=N_G(v)\cup\{v\}\).
When there is no ambiguity,
we drop the subscript~$G$.
An \emph{independent set}
is a set of pairwise nonadjacent vertices.

For a vertex subset~\(S\subseteq V\),
the \emph{subgraph~\(G[S]\) of~\(G\) induced by~\(S\)}
is the graph with vertex set~\(S\)
and edge set~\(\{\{u,v\}\in E\mid u\in S, v\in S\}\).
We say that $G$~is \emph{\(F\)-free}
for some graph~\(F\)
if \(G\)~does not contain induced subgraphs
isomorphic to~\(F\).
A~\(K_{n,m}\)
is a complete bipartite graph
with \(n\)~vertices on the one side
and \(m\)~vertices on the other.

A graph is~$c$\hyp colorable
if one can assign each vertex one of \(c\)~colors
so that the endpoints of each edge have distinct colors.

We study the parameterized complexity of the following problem.
\problemdef{\boldmath\McCS{}}
{A  graph~$G=(V,E)$ with vertex weights
  $w \colon V \to \mathds{N}$.}
{Find a set~$S\subseteq V$ of maximum weight~\(w(S)\)
  such that \(G[S]\) is \(c\)\hyp colorable.}

\noindent
For \(c=1\),
we obtain the classic \MIS{} problem.
\problemdef{\MIS{}}
{A graph~$G=(V,E)$ with vertex weights
  $w \colon V \to \mathds{N}$.}
{Find a set~$S\subseteq V$
  of pairwise nonadjacent vertices
  and maximum weight~\(w(S)\).}
\noindent
If all vertices have weight one,
then we call the problems \cCS{} and \IS{},
respectively.

\paragraph{Graph classes.}
The motivation, applications, and recognition complexity
of the following graph classes
are discussed in detail in \cref{sec:classes}.
We first list a number of classic graph classes and then introduce the central graph class of this work, inductive $k$-independent graphs (\cref{def:iki}).

We use \(A\cap B\) to denote the class of graphs
that belong to both classes~\(A\) and~\(B\).

A graph is \emph{Hamiltonian}
if it contains a cycle using each vertex of the graph exactly once.

A graph is \emph{cubic}
if each vertex has exactly three neighbors.

A graph is \emph{chordal}
if it does not contain cycles
of length at least four as induced subgraphs.

A graph is an \emph{interval graph}
if its vertices can be represented
as intervals of the real line
such that two vertices
are adjacent
if and only if their corresponding intervals intersect.
If the vertices can be represented
by intervals of equal length, then
the graph is a \emph{unit interval graph}.

\begin{definition}[inductive \boldmath\ind{k}]
  \label[definition]{def:iki}
  A graph~$G$ is
  \emph{inductive \ind{k}}
  if there is a
  \emph{\indc{k} ordering}
  \(v_1,v_2,\ldots,v_n\)
  of its vertices such
  that all independent sets
  of~\(G[N[v_i] \cap \{v_i,v_{i+1},\ldots,v_n\}]\)
  have size at most~\(k\)
  for each~\(i\in\{1,2,\ldots,n\}\).
\end{definition}
\citet{YB12} survey
inductive \ind{k} graphs.

When dropping
the order constraint
from the definition of
inductive \ind{k} graphs,
that is,
if one requires the neighborhood
of each vertex to contain
independent sets of size at most~\(k\),
then one obtains the
$K_{1,k+1}$-free graphs---a~proper subclass.
\begin{definition}[\boldmath\(K_{1,k}\)-free graphs]
  \label[definition]{def:kclawfree}
  A graph is \emph{\(K_{1,k}\)-free}
  if it does not contain a \(K_{1,k}\)
  (a tree with one internal node and \(k\)~leaves)
  as an induced subgraph.
\end{definition}

\noindent
Another proper subclass
of inductive \ind{k} graphs
are  %
\emph{\(k\)\hyp simplicial graphs} \citep{YB12},
also studied by e.\,g.\;\citet{JM00}, \citet{KTV10}, and \citet{HT15}.

\begin{definition}[\boldmath\(k\)\hyp simplicial]
  \label[definition]{def:simplicial}
  A graph is \emph{\(k\)\hyp simplicial}
  if there is an ordering of its
  vertices~\(v_1,v_2,\ldots,v_n\)
  such that \(G[N[v_i]\cap\{v_i,v_{i+1},\ldots,v_n\}]\)
  can be partitioned into at most \(k\)~cliques
  for each~\(i\in\{1,2,\ldots,n\}\).
\end{definition}
When, again, dropping the order constraint,
one gets \emph{\(k\)-minoes} as
introduced by \citet{MT03},
which are a proper subclass of
\(K_{1,{k+1}}\)-free graphs:
\begin{definition}[\boldmath\(k\)-mino]
  \label[definition]{def:mino}
  A graph is called a \emph{\(k\)-mino} if each vertex is contained in
  at most \(k\)~maximal cliques. %
\end{definition}

\noindent
In several scheduling works (see \cref{sec:classes}),
one encounters variants
of \MIS{} in graphs of the following form:\footnote{We are not aware that \(A\bowtie B\) or an alike notation has been used before. Rather, previous work has been coming up with ad-hoc names for the classes \(A\bowtie B\) for various~\(A\) and~\(B\), often referring to one and the same class by several names.}
\begin{definition}[$A\bowtie B$]
  For two graph classes \(A\) and \(B\),
  we denote by \(A\bowtie B\) the class of
  graphs~\(G=(V,E)\) such that \(E=E_1\cup E_2\)
  for a graph~\(G_1=(V,E_1)\) in class~\(A\) and
  a graph~\(G_2=(V,E_2)\) in class~\(B\).
\end{definition}

\begin{figure*}
  \begin{subfigure}[b]{5cm}
    \begin{center}
      \includegraphics[scale=0.75]{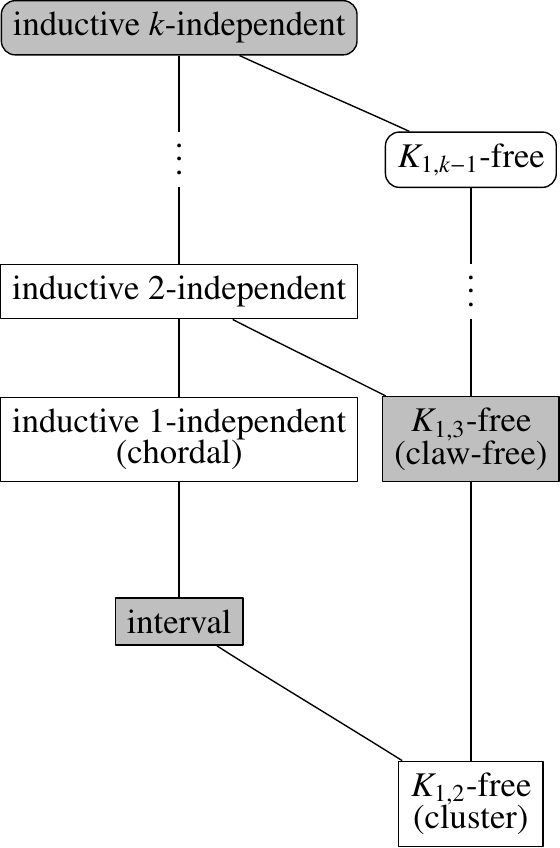}
    \end{center}
    \caption{The backbone, described in \cref{subsec:backbone}.}
    \label{fig:gc-backbone}
  \end{subfigure}\hfill
  \begin{subfigure}[b]{11cm}
    \includegraphics[scale=0.75]{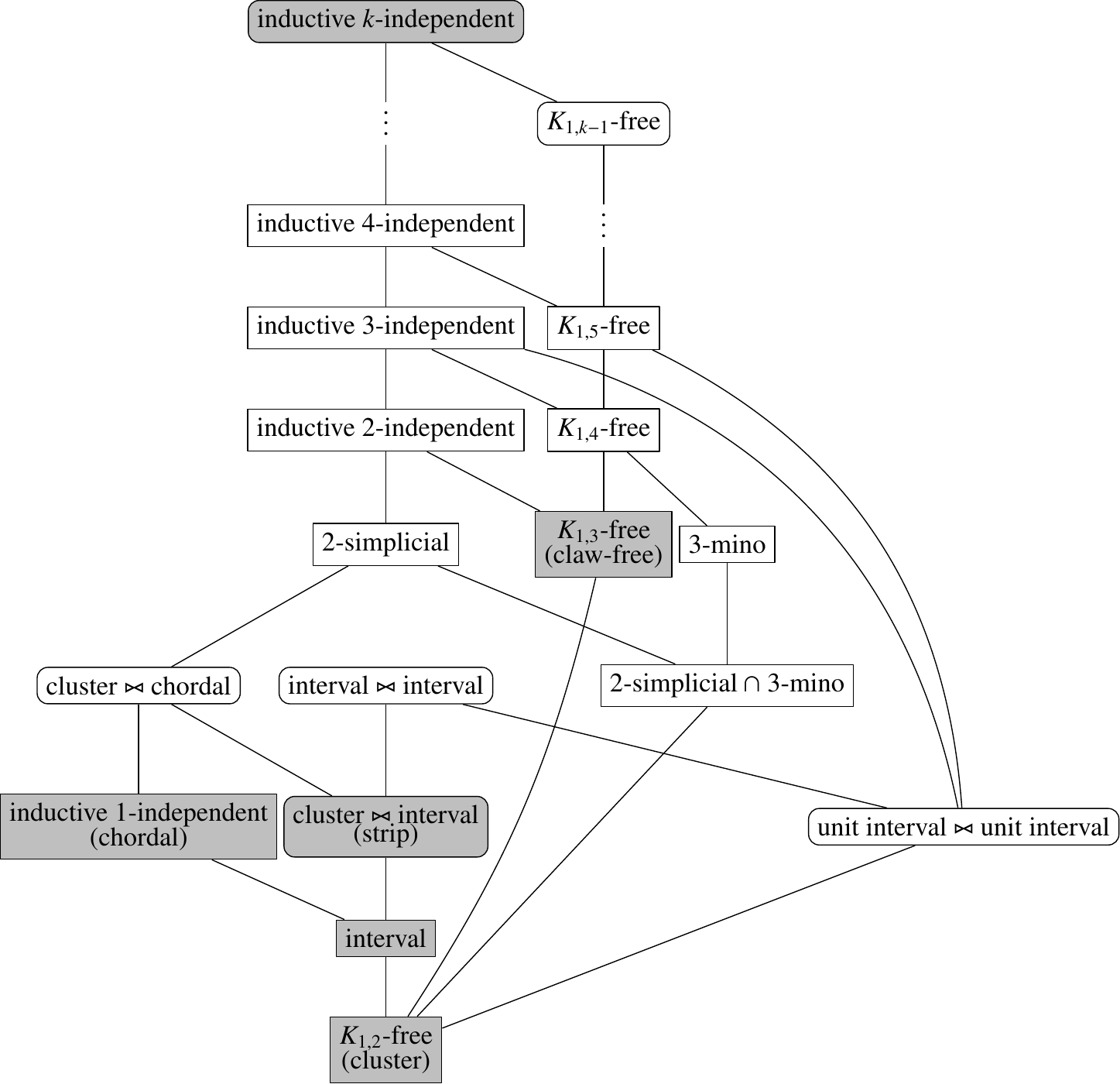}
    \caption{The enriched picture, described in \cref{subsec:enrich}.}
    \label{fig:gc-refined}
  \end{subfigure}
  \caption{Graph classes discussed in our work.
    The grey graph classes have been studied
    in the scheduling literature before.
    Rectangular boxes show graph classes recognizable in polynomial time,
   boxes with rounded corners show graph classes whose recognition is NP-hard or W[1]-hard.
    An edge from a lower graph class to a higher graph class means that the lower one is properly contained in the upper one.}
\end{figure*}
\section{Graph classes related to scheduling}
\label{sec:classes} 
Before delving deeper into problem\hyp specific considerations, in this section
we collect a number of mostly simple properties of and relations between graph classes in the range from 
cluster graphs (a special case of interval graphs) 
on the ``bottom end'' and inductive \ind{2} graphs on the
``top end''.
In this way, we lay the foundations for subsequent studies 
mainly related to independent set problems (with their connections to scheduling).
More specifically, we exhibit the inclusion relationships 
between various graph classes, discuss the time complexities of their 
recognition problems (that is, given a graph, how
hard is it to decide whether it belongs to the particular graph class),
and indicate whether the graph class has already been 
investigated in the context 
of scheduling.
This section might also be of independent
interest to researchers outside of scheduling contexts. 

First, we will discuss some kind of backbone structure
of graph classes
(\cref{subsec:backbone}) and then 
we will enrich the picture by  further
natural graph
classes (\cref{subsec:enrich}).

\subsection{Backbone graph classes}
\label{subsec:backbone}

To start with, consider \cref{fig:gc-backbone},
which depicts the fundamental
backbone structure of our graph classes. 
Basically, it shows the relations between various levels of inductive 
\ind{k} and $K_{1,k}$-free graphs, enriched by interval graphs.
 Notably, here and also in \cref{fig:gc-refined},
 all shown containments 
 are proper. 
Next, we individually discuss the most relevant graph classes.

\paragraph{Interval graphs.} 
Interval graphs certainly
are among
the most fundamental graph 
classes studied in the context
of scheduling.
Among other things,
they are used to model
\emph{interval scheduling} problems,
where not only the lengths
but also the start times of jobs are given
and the task is
to minimize the number of required machines
or to select most profitable set of jobs \citep{KLPS07}.
Other uses include
the problem of
maximizing the weighted number of
\emph{just-in-time jobs}
(exactly meeting their deadlines)
on parallel identical machines
\citep{SV05},
\emph{mutual exclusion scheduling}
\citep{BC96,Gar09},
or selecting the optimal start times
of jobs from finite sets
\citep{BMNW15}.
Interval graphs
can be recognized in linear time \citep{BL76}, contain 
cluster graphs,
and are special cases of chordal graphs \citep{Dir61},
which are the same as inductive \ind{1} graphs. 

\paragraph{Cluster graphs.} 
These are disjoint unions of cliques (often occurring in the context of
graph-based data clustering \citep{BBC04,SST04}
and are special cases of unit interval graphs.
They are exactly the $K_{1,2}$-free graphs (where $K_{1,2}$ is a path on
three vertices). It is easy to see 
that they can be recognized in polynomial time.
In scheduling,
cluster graphs naturally
occur as conflict graphs
when,
for example,
modeling several variants
of one job
as multiple vertices
and only one of them
is allowed to be in a solution
(for example,
a job might be available
at various starting times
or
with various processing times,
which influences its profit or weight)
\citep{HK06,BMNW15}.

\paragraph{Chordal graphs.}
Chordal graphs are equivalent to inductive \ind{1} graphs
\citep{AADK02,BP93,YB12}.
They are a well-known superclass of interval graphs \citep{Dir61}.
They can be recognized in linear time \citep{RTL76}
and are applied in modeling
throughput\hyp maximization scheduling
problems in wireless networks \citep{BCR+12}.
Moreover,
the problem of scheduling jobs with unit execution times
and precedence constraints
on parallel identical machines
is polynomial\hyp time
solvable if the incomparability graph
of the partial order
determining the precedence constraints
is chordal \citep{PY79}.

\paragraph{Claw-free graphs.}
Claw-free graphs,
that is, \(K_{1,3}\)-free graphs,
trivially contain cluster graphs and are incomparable with 
both interval and chordal graphs
(a $K_{1,3}$ is an interval graph and a chordal graph,
yet an induced \(4\)\hyp cycle is \(K_{1,3}\)-free
but not chordal and thus not an interval graph).
It is easy to see that they can be recognized in polynomial time. 
Claw-free graphs appear in the context of wireless 
scheduling
\citep{KEGGM17,KM17} as well as in classic scheduling 
scenarios \citep{GK03}.

\paragraph{Inductive \ind{k} graphs for $k\geq 2$.}
These graphs recently gained increased interest in the context of wireless
scheduling \citep{AHT17,HT15} and clearly contain chordal graphs
and claw-free graphs (both by definition).

There is a straightforward polynomial\hyp time 
recognition algorithm;
however, the degree of the polynomial depends
on~$k$ and thus the algorithm is impractical
already for small~$k$.
On the one hand,
this is unfortunate
in view of the fact that
algorithms for inductive \(k\)\hyp independent graphs,
for example the approximation algorithm
for \McCS{} suggested
by \citet{YB12},
require the \indc{k} ordering to be known.
On the other hand,
like in the wireless scheduling
applications of \citet{AHT17} and \citet{HT15},
a sufficiently good \indc{k} ordering
is given directly by the application data
and does not have to be computed.

\looseness=-1
As observed by \citet{YB12},
having the degree in the polynomial
of the running time of the recognition algorithm
for inductive \(k\)\hyp independent graphs
depending on~\(k\) is unavoidable
unless FPT${}={}$W[1].
Since \citet{YB12} omitted the formal proof,
for the sake of completeness we provide it here.

\begin{proposition}\label[proposition]{prop:indkind-hard}
  Deciding whether a graph is inductive \ind{k}
  is coW[1]-hard with respect to~$k$.
\end{proposition}
\begin{proof}
  We reduce any instance~$(G=(V,E), k+1)$ of \IS{}
  to an instance~$G'=(V',E')$ such that
  $G$~contains an independent set of size~$k+1$
  if and only if
  $G'$~is \emph{not} inductive \ind{k}.
  Since \IS{} is W[1]-hard parameterized
  by~\(k\),
  this shows that  
  recognizing \emph{non}\hyp inductive \ind{k}
  graphs is W[1]-hard.
  Thus,
  recognizing inductive \ind{k} graphs
  is coW[1]-hard.

  The reduction works as follows.
  Let \(V=\{v_1,v_2,\ldots,v_n\}\).
  Then,
  $V':=V\cup\{u_1,u_2,\ldots, u_{k+1}\}$
  for \(k+1\)~new vertices~\(u_1,u_2,\allowbreak\ldots,\allowbreak u_{k+1}\)
  and
  $E':=E\cup\{\{u_i,v_j\} \mid 1 \leq i \leq k+1, 1 \leq j \leq n\}.$
  This completes our reduction.

  Now,
  assume that
  a maximum independent set in~$G$ has size at most~$k$.
  We show that
  $G'$~is inductive \ind{k}.
  Put each of the newly introduced vertices~$u_i$ for \(i\in\{1,2,\ldots,k+1\}\) first within the \indc{k} ordering and then all vertices~$v\in V$ in an arbitrary order.
  Since~$N_{G'}(u_i) = V$, the neighborhood of any vertex (ignoring vertices that come before it in the \indc{k} ordering) induces a subgraph of~$G$.
  Thus it trivially holds that this neighborhood only contains independent sets of size at most~$k$. 

  Now,
  assume that
  $G$~contains an independent set of size at least~$k+1$.
  Then,
  $G'$~is not inductive~$k$-independent:
  no vertex~\(u_i\) for \(i\in\{1,2,\ldots,k+1\}\)
  can be the first in a \indc{k} ordering
  since \(N(u_i)=V\).
  Moreover,
  no vertex~\(v_j\) for \(j\in\{1,2,\ldots,n\}\)
  can be the first in an \indc{k} ordering
  since $N(v_j)\supseteq\{u_1,u_2,\ldots,u_{k+1}\}$,
  which are pairwise nonadjacent.
\qed
\end{proof}

\paragraph{$K_{1,k}$-free graphs for $k\geq 4$.} 
These are obvious superclasses of claw-free graphs and, by definition, each $K_{1,k}$-free graph is inductive \ind{(k-1)}.
Again, there is a straightforward polynomial\hyp time recognition algorithm; however, the degree of the polynomial depends on~$k$ and thus the algorithm is impractical already for small~$k$.
A folklore reduction
shows that this is unavoidable unless FPT${}={}$W[1].

\begin{proposition}
Recognizing~$K_{1,k}$-free graphs is coW[1]-hard with respect to~$k$.
\end{proposition}
\begin{proof}
  We will give a parameterized reduction from \IS{},
  which is W[1]-hard parameterized by solution size~$k$.
Let~$(G=(V,E),k)$ be an instance of \IS{}.
We construct a graph~$G'$
that is \emph{not} $K_{1,k}$-free
if and only if $G$~contains an independent set of size~$k$.

Let \(u\)~be a vertex not in~\(V\)
and $G' = (V \cup \{u\}, E \cup \{\{u,v\}\mid v\in V\})$.
If $G$~contains an independent set~$I$ of size at least~$k$,
then~$G'[I \cup \{u\}]$ is a~$K_{1,k}$.
Thus,
\(G'\)~is not \(K_{1,k}\)-free.

Now, assume that $G$~does not
contain an independent set of size at least~$k$.
Then,
$G'$ does not contain such an independent set
either since \(u\)~is adjacent to all other vertices of~\(G'\).
Thus,
no vertex in~$G'$ can have $k$~pairwise nonadjacent neighbors
in~$G'$ and, thus, $G'$ is~$K_{1,k}$-free.
\qed
\end{proof}

\subsection{Further graph classes}
\label{subsec:enrich}

Now, we present a number of graph classes that enrich our backbone 
structure;
some of the new graph classes have prominently
appeared in the context of scheduling. Adding these further graph 
classes to \cref{fig:gc-backbone} leads to the ``richer'' 
\cref{fig:gc-refined}.
We remark that even though we could not spot literature references 
for some of the subsequent graph classes in the context of 
scheduling, we advocate their relevance in the scheduling context because 
they naturally generalize or specialize
established ``scheduling graph classes''; 
thus, considering these classes may lead to strengthenings of some known results.

\paragraph{Strip graphs.}
These are equivalent to
the class of cluster $\bowtie$ interval graphs
and thus
form an obvious superclass of interval graphs.
Their recognition is NP-hard;
they find applications in classic 
scheduling \citep{HK06},
in particular in modeling the
\textsc{Job Interval Selection} problem
introduced by
\citet{NH82}.

\paragraph{Cluster $\bowtie$ chordal graphs.}
These graphs form an obvious superclass of strip graphs and of chordal graphs.
We will see that their recognition is NP-hard as well (under Turing reductions; see 
\cref{cor:cluster-chordal-hard}). We are not aware of a direct 
scheduling application.

We will now show that cluster\({}\bowtie{}\)chordal graphs form a proper subclass of 2\hyp simplicial graphs and then prove the hardness of recognizing cluster\({}\bowtie{}\)chordal graphs.

\begin{proposition}
\label[proposition]{prop:2-simplicial}
  Cluster\({}\bowtie{}\)chordal graphs are
  a proper subclass of 2\hyp simplicial graphs.
\end{proposition}

\begin{proof}
  Consider a cluster\({}\bowtie{}\)chordal graph~\(G\).
  We show that
  any of its induced subgraphs~\(G'\)
  contains a vertex
  whose neighborhood
  can be covered by two cliques.
  Repeatedly deleting
  such vertices from~\(G\)
  yields an ordering
  as
  required by \cref{def:simplicial}.
  
  Note that \(G'\), like its supergraph~\(G\),
  is also an
  edge-wise union
  of a chordal graph~\(G_1\)
  and a cluster graph~\(G_2\)
  on the same vertex set.
  Since \(G_1\)~is chordal,
  it contains a vertex~\(v\)
  whose neighborhood is a clique
  \citep{BP93}.
  Its neighborhood in~\(G_2\) is also a clique.
  Thus,
  the neighborhood of~\(v\) in~\(G'\)
  can be covered by two cliques. This concludes the proof that cluster\({}\bowtie{}\)chordal graphs form a subclass of 2\hyp simplicial graphs.
  
  Next we show that this inclusion is strict.
  To this end, we show that $K_{2,4}$~is 2\hyp simplicial but not a cluster\({}\bowtie{}\)chordal graph.
  A~$K_{2,4}$ is 2\hyp simplicial
  since every induced subgraph
  contains a vertex
  whose neighborhood can be covered
  using two cliques.
  It is,
  however,
  not cluster\({}\bowtie{}\)chordal:
  if \(K_{2,4}\)~was the edge-wise union of a
  chordal graph~\(G_1\)
  and a cluster graph~\(G_2\),
  then
  \(G_2\) would contain at most two edges of~\(K_{2,4}\).
  Independently
  of the choice of these two edges,
  the remaining edges
  contain an induced~\(C_4\)
  and thus cannot be part of~\(G_1\).
\qed
\end{proof}

\noindent
Next,
we show that recognizing cluster\({}\bowtie{}\)chordal graphs
is NP-hard.
We will basically use the same reduction 
from the NP-hard \HC{} problem on triangle\hyp free cubic graphs
that \citet{HK06}
used to show that recognizing strip graphs is NP-hard.
Yet we need to adapt their proof, particularly concerning the Hamiltonicity of triangle-free cubic graphs (\cref{lem:hamiltoniscc}).
For the sake of self-containedness, we provide all details, including also parts that coincide with \citet{HK06}.

\decproblemdef{\HC}%
{A graph~\(G=(V,E)\)}%
{Is there a cycle in~$G$ that passes through every vertex exactly once?}

\noindent
We start with an auxiliary lemma.
\begin{lemma}
\label[lemma]{lem:hamiltoniscc}
A triangle\hyp free cubic graph~$G=(V,E)$ is Hamiltonian if and only if, for each vertex~$v$,
there is an edge~$e$ incident to~$v$ such that~$G'=(V,E\setminus\{e\})$ is a cluster\({}\bowtie{}\)chordal graph. 
\end{lemma}

\begin{proof}
``\(\Rightarrow\)'' Assume that $G$~contains a Hamiltonian cycle~$C = (c_1,c_2,\ldots, c_n)$.
Then, $\{c_i,c_{i+1}\} \in E$ for~$i \in \{1,2,\ldots n-1\}$ and~$\{c_n,c_1\}\in E$.
Without loss of generality,
let $v = c_n$.
We will prove that~$E\setminus \{c_n,c_1\}$ can be partitioned into two sets~$E_1$ and~$E_2$ such that~$(V,E_1)$ is a cluster graph and~$(V,E_2)$ is chordal.
First,
we define~$E_2 = \{\{c_i,c_{i+1}\} \mid 1 \leq i \leq n-1\}$.
The graph~$(V,E_2)$ is  a path and, hence, chordal.
Now,
observe that each vertex
has degree exactly one in~$G_1$.
Thus,
$G_1$~is a perfect matching and, thus,
a cluster graph.

``\(\Leftarrow\)''
Assume that,
for an arbitrary but fixed vertex~$v$,
there is an edge~$\{v,w\}$
such that~$G'=(V,E\setminus\{\{v,w\}\})$
is a cluster\({}\bowtie{}\)chordal graph.
We will show that~$G$ contains a Hamiltonian cycle.
Fix two graphs~$G_1=(V,E_1)$ and~${G_2=(V,E_2)}$ such that~$E_1 \cup E_2 = E$, where
$G_1$~is a cluster graph
and~$G_2$ is a chordal graph.
Since~$G$ is triangle\hyp free,
it holds that~$G_1$ and~$G_2$ are triangle\hyp free, too.
Every vertex has degree at most one in~$G_1$
since, otherwise, there would be a connected component
of size at least three in~$G_1$,
which is a contradiction to the fact
that~$G_1$ is a triangle\hyp free cluster graph.
Since~$G$ is a cubic graph,
it holds that each vertex except for~$v$ and~$w$
has degree at least two in~$G_2$.
Since~$G_2$ is chordal,
it has an \indc{1} ordering,
that is,
for each vertex~$u$,
the set of succeeding neighbors of~$u$ induce a clique.
Thus,
each vertex has at most one succeeding neighbor;
otherwise, $G_2$~would contain a triangle.
Since there is an ordering of the vertices of~\(G_2\)
such that each vertex has at most one succeeding neighbor
and only~$v$ and~$w$ can have degree one in~$G_2$,
we get that~$G_2$ is a path through all vertices
with~$v$ and~$w$ as endpoints.
Adding~$\{v,w\}$ to this path
gives a Hamiltonian cycle for~$G$.
\qed
\end{proof}

\noindent
Using \cref{lem:hamiltoniscc},
we can now easily show the following.

\begin{proposition}
\label[proposition]{cor:cluster-chordal-hard}
Recognizing cluster\({}\bowtie{}\)chordal graphs is NP-hard.
\end{proposition}
\begin{proof}
  Assume that cluster\({}\bowtie{}\)chordal graphs
  are recognizable in polynomial time.
  Assuming this, we derive a polynomial\hyp time algorithm
  for the NP-hard problem of checking whether a triangle\hyp free cubic graph~\(G\)
  is Hamiltonian \citep{WS84}, which implies P\({}={}\)NP.

  Let \(v\)~be any vertex of~\(G\).
  Since~$G$ is cubic,
  there are three edges incident to~$v$ in~$G$.
  If $G$~is Hamiltonian,
  then,
  by \cref{lem:hamiltoniscc},
  deleting one out of the three edges will
  turn~$G$ into a cluster\({}\bowtie{}\)chordal graph.
  Our algorithm simply tries out all three possibilities and,
  for each of them,
  checks whether we get a cluster\({}\bowtie{}\)chordal
  graph in polynomial time.
  \qed
\end{proof}

\paragraph{2\hyp simplicial graphs.}
These graphs form a superclass of cluster $\bowtie$ chordal 
graphs (\cref{prop:2-simplicial})
and, by definition, are inductive \ind{2}.

They can be easily recognized in 
polynomial time by the following algorithm:
as long as possible,
find and delete a vertex whose
neighborhood can be covered by two cliques,
which is equivalent to checking
whether the complement of the neighborhood is 2\hyp colorable.
The input graph is 2\hyp simplicial
if and only if this process terminates
with the empty graph as result.
We are not aware of a direct 
scheduling application.

\paragraph{3-minoes.}
In 3-minoes every vertex is contained in at most three maximal cliques.
It is straightforward to see that they form a subclass of $K_{1,4}$-free
graphs; they can be recognized in polynomial time \citep{MT03}.
We are not aware of a direct
scheduling application.

\paragraph{Interval\({}\bowtie{}\)interval graphs.}
This graph class is also known as 2-track interval graphs
and as 2-union graphs.
They trivially contain interval graphs.
The corresponding recognition
problem is NP-hard \citep{GW95}.
They have strong scheduling applications
in industrial %
steel manufacturing \citep{HKML11}.

\paragraph{Unit interval\({}\bowtie{}\)unit interval graphs.}
These are special cases of 2-track interval graphs and are also called 
2-track unit interval graphs.  They are contained in the class 
of inductive \ind{3} graphs \citep{YB12} and in the 
class of $K_{1,5}$-free graphs
because unit interval graphs are $K_{1,3}$-free \citep{Weg61}.
Their recognition problem is NP-hard \citep{Jia13}.
We are not aware of a direct
scheduling application.

\section{W[1]-hardness on 2-simplicial 3-minoes}
\label{sec:ishard}
\IS{} 
parameterized by the solution size~\(\ell\) is W[1]-hard
on inductive \ind{3} graphs;
\citet{Jia10} showed this
on unit interval\({}\bowtie{}\)unit interval graphs,
which are a subclass of inductive \ind{3} graphs \citep{YB12}.

In contrast, the more general
\cCS{} problem parameterized by the solution size~\(\ell\) is \fp{} tractable on inductive \ind{1} (that is, chordal) graphs
\citep{MPR+13}.
In \cref{sec:iseasy},
we will generalize this tractability result
to
cluster\({}\bowtie{}\)chordal graphs,
a non\hyp chordal proper subclass
of inductive 2\hyp independent graphs (cf.\ \cref{fig:gc-refined}).

Strengthening the result by \cite{Jia10}, we show that \McCS{} in
inductive \ind{2} graphs is W[1]-hard.
More precisely,
we show that
already \IS{} 
parameterized by the solution size~\(\ell\)
is W[1]-hard
even on 2\hyp simplicial 3-minoes
(see \cref{def:simplicial,def:mino}).
Since these are \(K_{1,4}\)-free
(see \cref{fig:MIS}),
our result also contrasts the fact
that \MIS{} is solvable in \(O(n^3)\)~time on
\(K_{1,3}\)-free graphs \citep{FOS11}.

\begin{theorem}\label{thm:ishard}
\IS{} parameterized by solution size~$\ell$ is W[1]-hard on 2\hyp simplicial 3-minoes.
\end{theorem}

\noindent
It may be tempting to prove \cref{thm:ishard}
using a parameterized reduction
from \IS{} on graphs of small degree; however
such efforts must be in vein since
\IS{} parameterized by solution size is \fp{} tractable on graphs
even with logarithmic degree
via a simple search tree algorithm.

To prove \cref{thm:ishard},
we use a parameterized reduction
from the \MCC{} problem,
which is W[1]-hard with respect to~\(k\) \citep{FHRV09}.

\problemdef{\MCC{}}
{A graph~\(G\) whose vertex set is partitioned into~$k$ independent sets~\(\Vrt{1}\uplus \Vrt{2}\uplus\dots\uplus \Vrt{k}\).}
{Does \(G\)~contain a complete subgraph of order~\(k\)?}
\begin{figure*}
  \centering
  \includegraphics{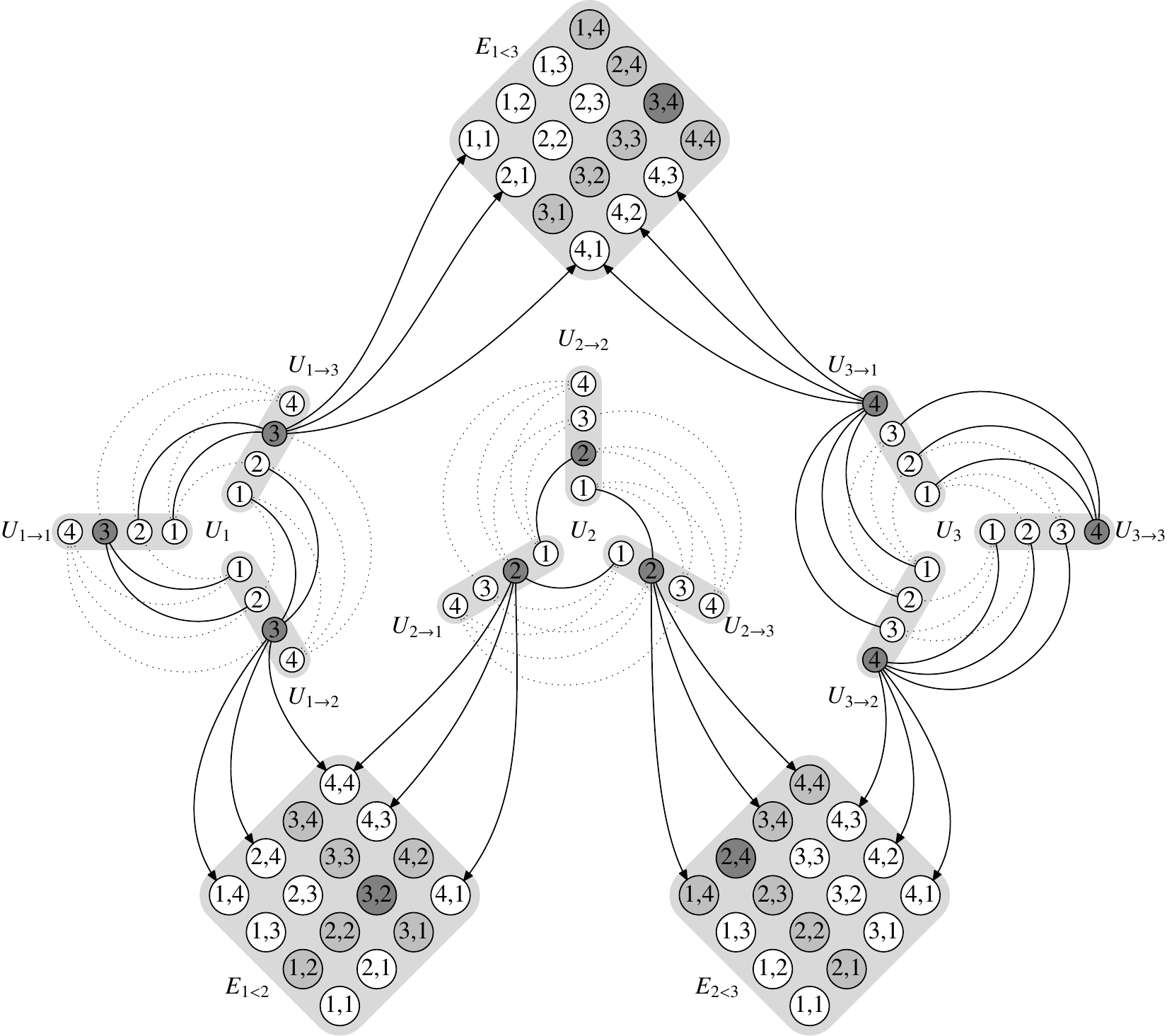}
  \caption{A \IS{} instance
    created by \cref{constr:ishard}
    from a \MCC{} instance
    with three color classes~\(\Vrt1,\Vrt2,\) and~\(\Vrt3\)
    of four vertices each,
    where all edges between vertices in~\(\Vrt i\)
    and~\(\Vrt j\)
    for \(i\ne j\) exist.
    Greyshaded areas are cliques.
    Dotted edges are ordinary edges and
    are dotted merely to unclutter the figure.
    The dark grey vertices form an independent set
    of maximum size (one vertex of each clique).
    Connections from vertex selection gadgets
    to verification gadgets~\(\Ert{i<j}\)
    are only shown for the dark grey vertices
    in \(\Urt{i\to j}\) and \(\Urt{j\to i}\):
    a directed arc means
    that the tail of the arc
    has edges to all vertices in the row or
    column pointed to by the head.
    Light gray vertices in the verification gadgets~\(\Ert{i<j}\)
    are those blocked by only one arc
    pointing at~\(\Ert{i<j}\).
    The labels of the vertices indicate their
    superscript in the corresponding clique
    used in \cref{constr:ishard}.}
  \label{fig:ishard}
\end{figure*}

\noindent
We will also refer to the sets \(V_i\)
for \(i\in\{1,2,\ldots,k\}\)
as \emph{color classes}.
We now describe our parameterized reduction
from \MCC{} to \IS{},
which is illustrated in \cref{fig:ishard}.
Subsequently, we prove its correctness
and that it creates 2\hyp simplicial 3-minoes.

\begin{construction}\label[construction]{constr:ishard}
  Given a \MCC{} instance~\(G=(V_1\uplus V_2 \uplus\dots\uplus V_k,E)\),
  we create an instance~\(G'\) of \IS{}
  that has a solution
  of size~\(\ell:=k^2+{k\choose 2}\)
  if and only if \(G\)~has a clique of size~\(k\).
  
  For each color class~\(\Vrt i=\{\vrt{i}{(1)},
  \vrt{i}{(2)},\ldots,
  \allowbreak
  \vrt{i}{(n_i)}\}\) of~\(G\),
  graph~\(G'\)
  contains a \emph{vertex selection gadget}~\(\Urt i\),
  which
  consists of \(k\)~cliques
  \(\Urt{i\to j}\),
  where \(j\in\{1,2,\ldots,k\}\).
  For each \(j\in\{1,2,\ldots,\allowbreak k\}\),
  clique~\(\Urt{i\to j}\)
  will be used to connect
  gadget~\(\Urt i\) to gadget~\(\Urt j\)
  and
  consists of
  newly introduced vertices~\(\{\urt{i\to j}{(1)},
  \allowbreak \urt{i\to j}{(2)},\ldots,
  \allowbreak \urt{i\to j}{(n_i)}\}\),
  each vertex~\(\urt{i\to j}{(p)}\)
  with \(p\in\{2,3,\ldots,n_i\}\)
  of which
  is adjacent
  to each vertex~\(\urt{i\to l}{(q)}\)
  with \(q\in\{1,2,\ldots,p-1\}\)
  and \(l = (j\bmod k) + 1\).

  Furthermore,
  for each two color
  classes~\(\Vrt i=\{\vrt{i}{(1)},\vrt{i}{(2)},
  \allowbreak\ldots,
  \allowbreak \vrt{i}{(n_i)}\}\)
  and
  \(\Vrt j=\{\vrt{i}{(1)},\vrt{j}{(2)},\ldots,
  \allowbreak \vrt{j}{(n_j)}\}\) of~\(G\)
  such that
  \(1\leq i< j\leq k\),
  graph~\(G'\) contains a
  \emph{verification gadget}~\(\Ert{i<j}\),
  which
  is a clique
  on the vertices~\(\ert{i< j}{p,q}\),
  newly
  introduced for each
  edge \(\{\vrt{i}{(p)},\vrt{j}{(q)}\}\in E\) of~\(G\).

  We connect the gadgets so that
  choosing a vertex~\(\ert{i<j}{p,q}\)
  of a verification gadget~\(\Ert{i<j}\)
  into an independent set
  enforces that \(\vrt{i}{(p)}\) of \(\Vrt{i}\)
  and \(\vrt{j}{(q)}\) of \(\Vrt{j}\)
  are also in the independent set,
  where \(i\)~is the smaller color index
  (which is reflected in the naming convention).
  To this end,
  for each edge~\(\{\vrt{i}{(p)},\vrt{j}{(q)}\}\in E\)
  with \(1\leq i<j\leq k\)
  of~\(G\),
  vertex~\(\urt{i\to j}{(p)}\) is adjacent
  to all vertices~\(\ert{i<j}{p',q'}\)
  such that \(p\ne p'\)
  and 
  \(\urt{j\to i}{(q)}\) is adjacent
  to all vertices~\(\ert{i<j}{p',q'}\)
  such that \(q\ne q'\).
  \qed
\end{construction}

\noindent
We now prove the correctness of \cref{constr:ishard}.
Thereafter,
it remains to show that it creates
2\hyp simplicial 3-minoes.

\begin{lemma}\label[lemma]{lem:iscor}
  A \MCC{} instance~\(G=(\Vrt1\uplus\Vrt2\uplus\dots\uplus \Vrt k,E)\)
  has a clique of order~\(k\)
  if and only if the graph~\(G'\) created by
  \cref{constr:ishard}
  has an independent set
  of size~\(\ell:=k^2+{k\choose 2}\).
\end{lemma}

\begin{proof}
  \((\Rightarrow)\)
  Let \(S\)~be a clique of order~\(k\) in~\(G\).
  It contains exactly one vertex of each color class.
  We describe an independent set~\(S'\)
  of size~\(\ell\)
  for~\(G'\).
  To this end,
  denote each color class~\(\Vrt i\) for
  \(i\in\{1,2,\ldots,k\}\) of~\(G\)
  as
  \(\Vrt i=\{\vrt{i}{(1)},\vrt{i}{(2)},\ldots,\allowbreak
  \vrt{i}{(n_i)}\}\).
  
  For each vertex~\(\vrt{i}{(p)}\)
  in~\(S\),
  set~\(S'\) contains
  the \(k\)~vertices \(\urt{i\to j}{(p)}\)
  with \(j\in\{1,2,\ldots,k\}\)
  of the vertex selection gadget~\(\Urt i\).
  Thus, \(S'\) contains
  exactly one vertex of
  each clique~\(\Urt{i\to j}\)
  and these vertices are pairwise nonadjacent:
  there are neither edges
  between vertices~\(\urt{i\to j}{(p)}\)
  and~\(\urt{i\to l}{(p)}\) for any \(j\ne l\) with
  \(\{j,l\}\subseteq\{1,2,\ldots,k\}\)
  nor are there edges between
  distinct vertex selection gadgets~\(\Urt i\)
  and~\(\Urt j\)
  in~\(G'\).

  For each edge~\(\{\vrt{i}{(p)},\vrt{j}{(q)}\}\)
  with \(1\leq i<j\leq k\)
  of~\(S\),
  set~\(S'\) contains
  the vertex~\(\ert{i<j}{p,q}\) of
  the verification gadget~\(E_{i<j}\).
  Since there are no edges
  between verification gadgets,
  these vertices are pairwise nonadjacent.
  Moreover,
  note that \(\ert{i<j}{p,q}\)
  has neither edges to \(\urt{i\to j}{(p)}\)
  nor to \(\urt{j\to i}{(q)}\),
  which are the only vertices
  of~\(\Urt{i\to j}\) and~\(\Urt{j\to i}\) in~\(S'\),
  nor has it edges to any vertex in~\(\Urt{i'\to j'}\)
  for \(\{i',j'\}\ne\{i,j\}\).
  Thus, \(S'\)~is an independent set.

  Finally,
  \(S'\)~has size
  \(\ell=k^2+{k\choose 2}\):
  it contains \(k\)~vertices
  for each of \(k\)~color classes of~\(G\)
  and one vertex
  for each edge of
  the clique~\(S\).

  \((\Leftarrow)\)
  Let \(S'\)~be an independent set
  of size~\(\ell\)
  for~\(G'\).
  We describe a clique
  of order~\(k\) in~\(G\).
  Since each vertex of~\(G'\)
  is in one of the
  \(\ell\)~pairwise vertex\hyp disjoint
  cliques~\(\Urt{i\to j}\)
  with \(\{i,j\}\subseteq\{1,2,\ldots,k\}\)
  and~\(\Ert{i<j}\)
  with \(1\leq i<j\leq k\),
  \(S'\)~contains exactly one vertex
  of each of them.
  Let \(S\)~be the set
  of vertices~\(\vrt{i}{(p)}\) of~\(G\)
  for \(i\in\{1,2,\ldots,k\}\)
  such that~\(\urt{i\to i}{(p)}\) of~\(\Urt{i\to i}\)
  is in~\(S'\).
  We will prove that the vertices in~\(S\)
  are pairwise adjacent in~\(G\).
  
  To this end,
  we first prove that \(p=q\)
  for any two vertices~\(\urt{i\to j}{(p)}\)
  and~\(\urt{i\to l}{(q)}\) in~\(S'\),
  where \(\{i,j,l\}\subseteq\{1,2,\ldots,k\}\).
  To this end,
  note that \(G'\)~contains
  edges
  from
  vertex~\(\urt{i\to j}{(p)}\in S'\)
  to each vertex~\(\urt{i\to l'}{(q')}\)
  with \(q'\in\{1,2,\ldots,p-1\}\)
  and \(l' = (j\bmod k) + 1\).
  Thus,
  for the vertex~\(\urt{i\to l'}{(q')}\in S'\),
  we have \(q'\geq p\).
  Iterating the argument,
  we get \(q\geq q'\geq p\), and,
  iterating further,
  \(p\geq q\geq q'\geq p\),
  and thus \(p=q\).

  We now prove that
  two arbitrary vertices~\(\vrt{i}{(p)}\)
  and~\(\vrt{j}{(q)}\) with \(1\leq i<j\leq k\) in~\(S\)
  are adjacent in~\(G\).
  By choice of~\(S\),
  we have that
  \(\urt{i\to i}{(p)}\) and \(\urt{j\to j}{(q)}\)
  are in~\(S'\).
  As shown in the previous paragraph,
  we thus also have that
  \(\urt{i\to j}{(p)}\) and \(\urt{j\to i}{(q)}\)
  are in~\(S'\).
  Since \(\urt{i\to j}{(p)}\) is adjacent
  to all vertices~\(\ert{i<j}{p',q'}\) of~\(\Ert{i,j}\)
  such that \(p\ne p'\)
  and
  \(\urt{j\to i}{(q)}\) is adjacent
  to all vertices~\(\urt{i<j}{p',q'}\)
  such that \(q\ne q'\),
  we get that \(S'\)~contains
  vertex~\(\ert{i<j}{p,q}\) of~\(\Ert{i<j}\).
  The existence of this vertex in~\(\Ert{i<j}\)
  shows that
  there is an edge~\(\{\vrt{i}{(p)},\vrt{j}{(q)}\}\in E\)
  in~\(G\).
  \qed
\end{proof}

\noindent
We will use the following lemma
to show that
\cref{constr:ishard} generates 2\hyp inductive 3-minoes
by showing that,
in any induced subgraph of the graphs
generated by \cref{constr:ishard},
we can find a vertex for which at
least one of the three sets
in the following lemma
are empty or contain only one vertex:

\begin{lemma}\label[lemma]{lem:tcover}
  Let \(\urt{i\to j}{(p)}\)
  with \(1\leq i\leq j\leq k\) be a vertex in
  the graph created by \cref{constr:ishard}
  from an instance~\(G=(\Vrt1\uplus\Vrt2\uplus \dots\uplus \Vrt k,E)\)
  of \MCC{}.
  
  Then,
  the neighborhood of \(\urt{i\to j}{(p)}\) can be
  covered by at most three cliques, consisting of
  \begin{enumerate}[i)]
  \item\label{set1} \(\{\urt{i\to j}{(q)}\mid q\geq p\}\cup\{\urt{i\to l}{(q)}\mid q< p, l=(j\bmod k)+1\}\),
    
  \item\label{set2} \(\{\urt{i\to j}{(q)}\mid q\leq p\}\cup \{\urt{i\to l}{(q)}\mid q> p, j=(l\bmod k)+1\}\), and
  \item\label{set3} 
    of \(\urt{i\to j}{(p)}\)
    and its neighbors in either~\(E_{i<j}\)
    or~\(E_{j<i}\),
    if \(i\ne j\).
  \end{enumerate}
\end{lemma}

\begin{proof}
  We first show that each
  of the sets \eqref{set1}--\eqref{set3} is a clique.
  This easily follows for \eqref{set3}
  since
  \(E_{i<j}\) for~\(1\leq i<j\leq k\)
  is a clique by \cref{constr:ishard}.

  We now prove that
  \eqref{set1} and \eqref{set2}
  are cliques.
  By \cref{constr:ishard},
  each vertex~\(\urt{i\to j'}{(p')}\)
  is adjacent
  to each vertex~\(\urt{i\to l'}{(q')}\)
  with \(q'\in\{1,2,\ldots,p'-1\}\)
  and \(l'=(j'\mod k)+1\).
  Herein,
  intuitively,
  we can think of~\(l'\) as the ``successor''
  to~\(j'\) in the cycle
  \((1,2,\ldots,k,1)\).
  By setting \(p':=p\) and \(j':=j\),
  we immediately get that \eqref{set1} is a clique.
  By choosing \(p':=p+1\) and
  \(l':=j\),
  we get that \eqref{set2} is a clique,
  since in \eqref{set2}
  the left set of the union now
  plays the role of the ``successor''
  of  the right set.

  It remains to prove that
  the sets \eqref{set1}--\eqref{set3}
  cover the whole neighborhood of~\(\urt{i\to j}{(p)}\).
  Vertex~\(\urt{i\to j}{(p)}\) has neighbors
  in all
  sets
  \(\Urt{i\to l}\)
  such that \(l=j\),
  \(l = (j\bmod k) + 1\),
  or \(j = (l\bmod k) + 1\).
  If \(i\ne j\),
  then it might also have neighbors in
  either~\(E_{i<j}\)
  or~\(E_{j<i}\).
  Thus,
  the union of the sets
  in \eqref{set1}--\eqref{set3}
  covers all
  neighbors in~\(E_{i<j}\),
  \(E_{j<i}\),
  and~\(\Urt{i\to j}\).
  Moreover,
  set~\eqref{set1} covers
  all neighbors in~\(\Urt{i\to l}\)
  for \(l = (j\bmod k) + 1\)
  by \cref{constr:ishard}.
  Finally,
  if set~\eqref{set2} did not cover
  all neighbors in~\(\Urt{i\to l}\)
  with \(j = (l\bmod k) + 1\),
  then this would mean that
  \(\urt{i\to j}{(p)}\)
  had a neighbor~\(\urt{i\to l}{(q)}\)
  with \(q\leq p\).
  This is impossible since,
  by \cref{constr:ishard},
  \(\urt{i\to l}{(q)}\) is
  adjacent to vertices~\(\urt{i\to j}{(p)}\)
  only for \(p<q\).  
  \qed
\end{proof}

\noindent
Using \cref{lem:tcover},
the following lemma is easy to prove
and finishes
the proof of \cref{thm:ishard}.

\begin{lemma}\label[lemma]{lem:isclass}
  \cref{constr:ishard} creates 2\hyp simplicial 3-minoes.
\end{lemma}

\begin{proof}
  We first show that \cref{constr:ishard}
  creates 3-minoes.
  By \cref{lem:tcover},
  the neighborhood of
  each vertex in vertex selection gadgets
  can be covered by three cliques.
  Also a vertex
  of a verification gadget~\(E_{i<j}\)
  has neighbors only in the
  three sets~\(E_{i<j}\),
  \(\Urt{i\to j}\), and \(\Urt{j\to i}\),
  each of which is a clique.

  We now show that the graphs~\(G\)
  created by \cref{constr:ishard}
  are 2\hyp simplicial.
  To this end,
  it is enough to show
  that each induced subgraph~\(G'\) of~\(G\)
  contains a vertex
  whose neighborhood
  can be covered by two cliques.
  Then,
  subsequently deleting a vertex
  whose neighborhood can be covered
  by two cliques
  gives an ordering
  as required by \cref{def:simplicial}.
  
  If \(G'\)~contains
  only vertices of
  verification gadgets,
  then the neighborhood
  of each vertex in~\(G'\)
  can be covered using one clique.
  If~\(G'\) contains a
  vertex~\(\urt{i\to i}{(p)}\) of~\(G\)
  for any~\(i\in\{1,2,\ldots,k\}\),
  then,
  by \cref{lem:tcover},
  its neighborhood
  can be covered by two cliques.
  Otherwise,
  there are \(j\in\{1,2,\ldots,k\}\)
  and
  \(l=(j\bmod k)+1\)
  such that
  \(G'\)~contains vertices
  of~\(\Urt{i\to j}\)
  but no vertices
  of~\(\Urt{i\to l}\).
  Let \(\urt{i\to j}{(p)}\)~be
  the vertex of~\(\Urt{i\to j}\) with
  maximum~\(p\) in~\(G'\).
  For this vertex,
  the set \eqref{set1} in \cref{lem:tcover}
  only contains~\(\urt{i\to j}{(p)}\)
  and is therefore contained in the other two sets.
  Thus,
  its neighborhood can be covered by two cliques---sets \eqref{set2} and \eqref{set3}.
  \qed
\end{proof}

\noindent
\cref{thm:ishard} now follows from
\cref{lem:isclass,lem:iscor},
the fact that \cref{constr:ishard}
runs in polynomial time,
and that \MCC{} is W[1]-hard
with respect to~\(k\) \citep{FHRV09}.

\section{Fixed\hyp parameter tractability on cluster\({}\bowtie{}\)chordal graphs}
\label{sec:iseasy}

In this section,
we prove that \McCS{}
parameterized by the solution size~\(\ell\) is \fp{} tractable
on
cluster\({}\bowtie{}\)chordal
graphs.
This complements
the negative result of \cref{sec:ishard}.

\begin{theorem}\label{thm:iseasy}
  A maximum\hyp weight \(c\)-colorable subgraph
  with at most \(\ell\)~vertices
  of a cluster\({}\bowtie{}\)chordal graph
  can be computed in \(2^{\ell+c}\cdot (c+e+3)^{\ell}\cdot \ell^{O(\log\ell)}\cdot n^2\cdot \log^3 n\)~time
  if the decomposition of the input graph
  into a cluster and a chordal graph is given.
  Herein, \(e\)~is Euler's number.
\end{theorem}

\noindent
Note that \(c\leq\ell\) holds in all nontrivial cases.
Thus,
\cref{thm:iseasy}
shows that \McCS{} parameterized by~\(\ell\)
is \fp{} tractable in cluster\({}\bowtie{}\)chordal graphs.
Moreover,
note that it also shows
that \McCS{} parameterized by the weight~\(W\)
of the sought solution is \fp{} tractable,
since if there is an independent set of weight~\(W\),
then there is also one
consisting of at most \(W\)~vertices.

On the one hand, \cref{thm:iseasy}
is a generalization of a result of \citet{MPR+13},
who showed
that \cCS{} parameterized by the solution size~\(\ell\) is \fp{} tractable on chordal graphs.
On the other hand,
it generalizes a
fixed\hyp parameter tractability result
on the \textsc{Job Interval Selection} problem
of \citet{BMNW15},
who showed
that \IS{}
parameterized by the solution size~\(\ell\)
is \fp{} tractable
on cluster\({}\bowtie{}\)interval graphs.

The proof of \cref{thm:iseasy}
works in three steps.
First,
in \cref{sec:mccs->mis},
we use the color coding technique
due to \citet{AYZ95}
to
reduce \McCS{} to \MIS{}.
Then,
in \cref{sec:mis->mcis},
we again use color coding
to
reduce \MIS{}
in cluster\({}\bowtie{}\)chordal graphs
to the problem of finding
a maximum\hyp weight independent set of pairwise
distinct colors in
chordal graphs
whose vertices are colored in \(\ell\)~colors.
Finally,
in \cref{sec:mcis-dp},
we use dynamic programming
to find this colorful independent set.

\subsection{From \McCS{} to \MIS{}}
\label{sec:mccs->mis}

In order to prove \cref{thm:iseasy},
we show that it is enough
to give a \fp{} algorithm for \MIS{}
parameterized by the solution size~\(\ell\).

To this end,
we use the color coding technique of \citet{AYZ95}.
This technique has been used by \citet{MPR+13}
to show that \cCS{} parameterized by the solution size~\(\ell\)
is \fp{} tractable in chordal graphs.
Their approach is neither limited
to the unweighted problem
nor to chordal graphs.
Note, however, that the approach requires some tuning
to work in the weighted setting.
We describe it here in the most general form.
To this end,
we need the following definition.

\begin{definition}[hereditary graph class]
  A graph class~\(\mathcal C\) is \emph{hereditary} if every
  induced subgraph of a graph in~\(\mathcal C\) also belongs
  to~\(\mathcal C\).
\end{definition}

\noindent
All graph classes
considered in this work are hereditary.

\begin{lemma}
  \label[lemma]{lem:mccs-to-is}
  Let \(\mathcal C\) be
  a hereditary graph class.
  Moreover,
  let \(t\colon\mathds{N}\times\mathds{N}\to\mathds{N}\)
  be a function, nondecreasing in both arguments,
  such that a maximum\hyp weight independent set
  of size at most~\(\ell\) in~\(\mathcal C\)
  can be computed in \(t(\ell,n)\)~time.

  Then,
  a maximum\hyp weight
  induced \(c\)\hyp colorable subgraph
  with at most \(\ell\)~vertices
  of a graph in \(\mathcal C\)~is computable
  in
  \(2^{\ell+c} c^\ell\ell^{O(\log\ell)}\log^2n\cdot t(\ell,n)\)~time.
\end{lemma}

\begin{proof}
  Our algorithm makes use
  of \emph{\((n,\ell,c)\)\hyp universal sets}
  \citep{MPR+13}---sets~\(\mathcal F\)
  of functions~\(f\colon \{1,2,\ldots,n\}\to\{1,2,\ldots,c\}\)
  such that,
  for any set~\(I\subseteq\{1,2,\ldots,n\}\)
  of size~\(|I|=\ell\)
  and
  any function~\(g\colon\{1,2,\ldots,n\}\to\{1,2,\ldots,c\}\),
  there is a function~\(f\in\mathcal F\)
  with \(f(i)=g(i)\) for each~\(i\in I\).

  Our algorithm computes
  a maximum\hyp weight \(c\)\hyp colorable subgraph~\(G[S]\)
  with at most \(\ell\)~vertices
  of a graph~\(G=(\{v_1,v_2,\allowbreak \ldots,\allowbreak v_n\},E)\)
  in~\(\mathcal C\)
  as follows.
  First,
  in \(c^\ell\ell^{O(\log \ell)}n\log^2n\)~time,
  compute an
  \((n,\ell,c)\)\hyp universal set~\(\mathcal F\)
  of size~\(|\mathcal F|\in  c^\ell\ell^{O(\log \ell)}\cdot \log^2n \)  \citep{MPR+13}.
  Then,
  iterate over each~\(f\in\mathcal F\)
  and each of the \(\binom{\ell+c}{c}\leq 2^{\ell+c}\)~possible
  partitions~\(\ell_1+\ell_2+\dots+\ell_c\leq \ell\)
  such that \(\ell_i\geq 0\)
  for each \(i\in\{1,2,\ldots,c\}\).
  In each iteration, do the following:
  Compute a maximum\hyp weight independent set~\(S_i\)
  of size at most~\(\ell_i\)
  in the subgraph~\(G_i\)
  induced by the vertices~\(v_k\)
  with \(f(k)=i\)
  for each \(i\in\{1,2,\ldots,c\}\).
  Return
  the maximum\hyp weight
  set~\(S'=S_1\uplus S_2\uplus\dots\uplus S_c\)
  found in any iteration.
  Note
  that the independent set~\(S_i\)
  in each graph~\(G_i\) for \(i\in\{1,2,\ldots,c\}\)
  can be found in \(t(\ell,n)\)~time
  since \(\mathcal C\)~is a hereditary graph class
  and, thus, \(G_i\in\mathcal C\).

  It is obvious that~\(G[S']\)
  has at most \(\ell\)~vertices
  and
  is \(c\)\hyp colorable
  since it is the union
  of \(c\)~independent sets.
  Thus, \(w(S')\leq w(S)\),
  where \(G[S]\)~is a maximum\hyp weight
  \(c\)\hyp colorable subgraph
  with at most \(\ell\)~vertices in~\(G\).

  It remains to show \(w(S')\geq w(S)\).
  Let \(\col\colon V\to\mathds{N}\)~be
  a proper \(c\)\hyp coloring of~\(G[S]\).
  Since \(|S|\leq\ell\),
  there is a function~\(f\in\mathcal F\)
  such that \(\col(v_k)=f(k)\) for each~\(v_k\in S\).
  For each~\(i\in\{1,2,\ldots,c\}\),
  let \(S_i:=\{v_k\in S\mid f(k)=i\}\)
  and
  \(\ell_i:=|S_i|\).
  Since \(G[S]\)~is a maximum\hyp weight \(c\)\hyp colorable
  subgraph,
  for each~\(i\in\{1,2,\ldots,c\}\),
  \(S_i\)~is a maximum\hyp weight independent set
  with at most \(\ell_i\)~vertices
  in the subgraph~\(G_i\)
  induced by the vertices~\(v_k\) with \(f(k)=i\).
  Since the algorithm iterated over this~\(f\)
  and these \(\ell_1+\ell_2+\dots+\ell_c\leq\ell\),
  we get that \(w(S')\geq w(S)\).
  \qed
\end{proof}

\noindent
We point out that,
for an easier and faster
implementation of the algorithm
in the proof of \cref{lem:mccs-to-is},
it makes sense to trade
the factor \(\ell^{(O\log \ell)}\log^2n\)
in the running time for a small error probability~\(\varepsilon\)
by simply choosing \(f(k)\in\{1,\dots,c\}\)
uniformly at random
for each vertex~\(v_k\) of~\(G\)
instead of using universal sets.
With probability at least \(c^{-\ell}\),
\(f\)~will be a proper coloring
for a maximum\hyp weight solution
and the algorithm will thus find it.
Repeating the algorithm
with
\(O(c^\ell\ln 1/\varepsilon)\)~random assignments~\(f\)
thus allows for finding a
maximum\hyp weight solution
with error probability at most~\(\varepsilon>0\).

\subsection{From \MIS{} to \MCIS{}}
\label{sec:mis->mcis}

In order to show \cref{thm:iseasy},
by \cref{lem:mccs-to-is},
it is enough to show that
\MIS{} parameterized by the solution size~\(\ell\) is \fp{} tractable in cluster\({}\bowtie{}\)chordal graphs.
We do this using an algorithm for
the following auxiliary problem.

\problemdef{\MCIS}%
{A graph~\(G=(V,E)\) with
  vertex weights~\(w\colon V\to\mathds{N}\)
  and
  a coloring~\(\col\colon V\to\{1,2,\ldots,c\}\) of its vertices.}%
{Find a maximum\hyp weight independent set whose vertices have pairwise distinct colors.}

\noindent
Colorful independent sets of
interval graphs have been studied
by \citet{BMNW15} in order to show that
\IS{} parameterized by the solution size~\(\ell\) is \fp{} tractable on cluster\({}\bowtie{}\)interval graphs,
thus strengthening a result
on the \textsc{Job Interval Selection}
problem of \citet{HK06}.

In \cref{sec:mcis-dp},
we will study \MCIS{} on chordal graphs.
Together with the following lemma,
this will give a \fp{} algorithm
for \MIS{} on
cluster\({}\bowtie{}\)chordal graphs.

\begin{lemma}
  \label[lemma]{lem:clu-to-co}
  Let \(\mathcal C\)~be a graph class and
  \(t\colon \mathds{N}\times\mathds{N}\to\mathds{N}\) be
  a function such that
  a maximum\hyp weight colorful independent set
  in graphs in~\(\mathcal C\)
  is computable in \(t(c,n)\)~time.

  Then, a maximum\hyp weight independent set
  of size at most~\(\ell\)
  in cluster\({}\bowtie \mathcal C\) graphs
  is computable
  in \(e^{\ell}{\ell}^{O(\log \ell)}\log n\cdot t(\ell,n)\)~time,
  assuming that the decomposition of the input graph
  into a cluster graph and a graph of~\(\mathcal C\)
  is given.
\end{lemma}

\begin{proof}
  Let \(G=(V,E_1\cup E_2)\)
  such that \(G_1=(V,E_1)\) is a cluster graph
  and \(G_2=(V,E_2)\in\mathcal C\).
  Let \(d\)~be the number of clusters in~\(G_1\).
  Our algorithm uses
  a \emph{\((d,\ell)\)\hyp perfect hash family}~\(\mathcal F\)---a~family
  of functions~\(f\colon\{1,2,\ldots,d\}\to\{1,2,\ldots,\ell\}\)
  such that,
  for each subset~\(C\subseteq\{1,2,\ldots,d\}\)
  of size~\(\ell\),
  at least one of the functions in~\(\mathcal F\)
  is a bijection between~\(C\) and~\(\{1,2,\ldots,\ell\}\).

  We find a maximum\hyp weight independent set
  with at most \(\ell\)~vertices in~\(G\)
  as follows.
  First,
  in \(e^{\ell}{\ell}^{O(\log \ell)}d\log d\)~time,
  construct
  a \((d,\ell)\)\hyp perfect hash family~\(\mathcal F\)
  with \(|\mathcal F|\in e^{\ell}{\ell}^{O(\log \ell)}\log d\)
  \citep[Definition 5.15--Theorem~5.18]{CFK+15}.
  Then,
  iterate over all functions~\(f\in\mathcal F\).
  In each iteration,
  consider the coloring~\(\col_f\)
  that
  colors the vertices
  in each cluster~\(i\in\{1,2,\ldots,d\}\)
  of~\(G_1\) using color~\(f(i)\in\{1,2,\ldots,\ell\}\)
  and compute a maximum\hyp weight colorful independent
  set in~\(G_2\) with respect to \(\col_f\)
  in \(t(\ell,n)\)~time.
  Return the colorful independent set~\(S'\)
  of maximum weight found in any iteration.

  We show that \(S'\)~is
  an independent set for~\(G\)
  with \(w(S')=w(S)\),
  where \(S\)~is a maximum\hyp weight independent set
  of size at most \(\ell\) in~\(G\).

  ``\(\leq\)'':
  Observe that \(S'\),
  being an independent set for~\(G_2\),
  contains no edges of~\(G_2\).
  Moreover,
  being colorful for some coloring~\(\col_f\),
  \(S'\) neither contains edges of~\(G_1\):
  the endpoints of such an edge
  belong to the same cluster~\(i\)
  and thus get the same color~\(\col_f(i)\).
  Thus,
  \(S'\)~is an independent set
  of size at most~\(\ell\) for~\(G\)
  and, therefore,
  \(w(S')\leq w(S)\).

  ``\(\geq\)'':
  Since \(S\)~uses at most one vertex
  of each cluster of~\(G_1\),
  the set~\(C_S\) of clusters
  that contain vertices of~\(S\)
  has size at most~\(\ell\).
  Thus,
  \(\mathcal F\)~contains
  a bijection~\(f\colon C_S\to\{1,2,\ldots,|C_S|\}\).
  Note that \(S\)~is
  a colorful independent
  set for~\(G_2\) with respect to~\(\col_f\).
  Thus, \(w(S')\geq w(S)\).
  \qed
\end{proof}

\noindent
As in the case with \cref{lem:mccs-to-is},
for an easier and faster implementation
of the algorithm
in the proof of \cref{lem:clu-to-co},
it makes sense to trade the factor \(\ell^{O(\log \ell)}\log d\)
in the running time
for a small error probability
by randomly coloring the clusters
instead of using perfect hash families:
the vertices of maximum\hyp weight solution
will have pairwise distinct colors with probability
at least \(\ell!/\ell^\ell\in\Theta(e^{-\ell})\).
Repeating the algorithm
with \(O(e^\ell\ln 1/\varepsilon)\) random
colorings will thus find 
a maximum\hyp weight solution
with error probability at most~\(\varepsilon>0\).

\subsection{\MCIS{} in chordal graphs}
\label{sec:mcis-dp}

To complete
our proof of \cref{thm:iseasy},
we finally show
the following proposition,
which,
together with
\cref{lem:mccs-to-is,lem:clu-to-co},
yields \cref{thm:iseasy}.

\begin{proposition}
  \label[proposition]{lem:mcis-chordal}
  \MCIS{} in chordal graphs
  is solvable in \(O(3^cn^2)\)~time.
\end{proposition}

\noindent
\cref{lem:mcis-chordal}
generalizes a result
of \citet{HK06},
who showed that \MCIS{} parameterized
by the number~\(c\) of colors is \fp{} tractable 
in interval graphs.
Whereas the algorithm of \citet{HK06}
is based on dynamic programming
over the interval representation
of interval graphs,
we will exploit
dynamic programming over
\emph{tree decompositions}
of chordal graphs:

\begin{definition}[tree decomposition]
  \label[definition]{def:td}
  A \emph{tree decomposition}
  of a graph~$G=(V,E)$ is a pair~$(\mathcal X,T)$,
  where $\mathcal X\subseteq 2^V$ and
  $T$~is a tree on~$\mathcal X$,  such that
  \begin{enumerate}[(i)]
  \item\label{td1} \(\bigcup_{X  \in \mathcal X} X  =V,\)   the sets~$X\in\mathcal X$ are called \emph{bags},
  \item\label{td2} for each edge~$e\in E$, there is a bag~$X\in \mathcal X$ with $e\subseteq X$,
  \item\label{td3} for any three
    bags~\(X,Y,Z\in\mathcal X\)
    such that \(Y\)~is on the path
    between~\(X\) and~\(Z\) in~\(T\),
    one has \(X\cap Z\subseteq Y\).
  \end{enumerate}
\end{definition}

\noindent
Chordal graphs have tree decompositions
in which each bag induces a clique
and such tree decomposition can be computed
in linear time \citep{BP93}.

We will additionally assume
that the tree decomposition is rooted
at an arbitrary bag
and
that each bag~\(X\) with more than one child in~\(T\)
has exactly two children~\(Y\) and~\(Z\) in~\(T\)
such that \(X=Y=Z\).
This property can be established step-by-step as follows:
if some bag~\(X\) has more than two children
or its children are not all equal to~\(X\),
then we can attach two copies~\(X_1\)
and~\(X_2\) of~\(X\)
to~\(X\) as children in~\(T\),
attach one original child of~\(X\) to~\(X_1\)
and the other children of~\(X\) to~\(X_2\).
This procedure,
in the worst case,
triples the number of bags in the tree decomposition
and works in linear time.

Using this linear\hyp time computable tree decomposition
and the following lemma with \(\alpha=1\)
yields \cref{lem:mcis-chordal},
which,
together with
\cref{lem:mccs-to-is,lem:clu-to-co},
concludes the proof of
\cref{thm:iseasy}.

\begin{lemma}
\label[lemma]{lem:td}
  Given a tree decomposition,
  each bag of which
  has at most \(\alpha\)~pairwise nonadjacent vertices,
  \MCIS{} can be solved
  in \(O(3^c\cdot n^{\alpha+1})\)~time.
\end{lemma}

\begin{proof}
  Let \(G=(V,E)\)~be a graph
  and let \((\mathcal X,T)\)
  be a tree decomposition
  such that each bag of~\(\mathcal X\)
  contains at most
  \(\alpha\)~pairwise nonadjacent vertices.
  Root $\mathcal{T}$ at an arbitrary bag~\(R\in\mathcal X\).
  For each bag~\(X\in\mathcal X\),
  let \(V_X\subseteq V\)~denote the set
  of vertices in~\(X\) and its descendant bags
  in~\(T\).
  Moreover,
  let \(G_X:=G[V_X]\)
  and \(\col(S):=\{\col(v)\mid v\in S\}\).

  For each bag~\(X\in\mathcal X\),
  each \(S\subseteq X\),
  each subset~\(C\subseteq\{1,2,\ldots,c\}\)
  of colors,
  let \(T[X,C,S]\)~denote
  the maximum weight
  of an
  independent set~\(I\) of~\(G_X\)
  with \(I\cap X=S\)
  whose vertices have
  pairwise distinct colors from~\(C\).
  Then,
  by \cref{def:td}\eqref{td1},  
  the solution to \MCIS{}
  is \(\max_{S\subseteq R} T[R,\{1,2,\ldots,c\},S].\)
  We now show how to compute it
  in \(O(3^cn^{\alpha+1})\)~time,
  considering each bag type separately.

  \paragraph{Infeasible solutions.}
  Since
  no bag has an independent set
  with more than \(\alpha\)~vertices,
  for each \(X\in\mathcal X\),
  each \(S\subseteq X\),
  and each \(C\subseteq\{1,2,\ldots,c\}\),
  we have
  \(T[X,C,S]=-\infty\)
  if \(|S|>\alpha\),
  or
  if \(S\)~is not independent,
  or if the vertices of~\(S\)
  do not have pairwise distinct colors from~\(C\).

  In the following,
  we will thus assume that
  \(S\subseteq X\)~is always an independent set and
  that its vertices have pairwise distinct colors
  from~\(C\subseteq\{1,2,\ldots,c\}\).
  
  \paragraph{Bags without children.}
  For each bag~\(X\in\mathcal X\)
  without children,
  \begin{gather}
    \label{eq:nochildren}
    T[X,C,S]=w(S).
  \end{gather}
  \paragraph{Bags with a single child.}
  Consider a bag~\(X\in\mathcal X\)
  that has one child~\(Y\).
  We show that
  \begin{gather}
    \label{eq:onechild}
    \begin{split}
    T&[X,C,S]=\\
     &\max
    \begin{Bmatrix*}[l]
      T[Y,C\setminus\col(S\setminus S'),S']+w(S\setminus S')\\
      \text{ for each independent set }S'\subseteq Y\\
      \text{\quad such that }S\cap X\cap Y=S'\cap X\cap Y
  \end{Bmatrix*}.
\end{split}
  \end{gather}
 Herein,
  the condition \(S\cap X\cap Y=S'\cap X\cap Y\)
  ensure that
  independent sets for~\(G_X\)
  and~\(G_Y\)
  agree on the vertices
  that are both in the bag~\(X\) and its child~\(Y\).
  Next, we show validity of~\eqref{eq:onechild}.
  
  ``\(\leq\)'': Consider a maximum\hyp weight
  independent set~\(I\) for~\(G_X\)
  with \(I\cap X=S\)
  whose vertices have pairwise distinct colors from~\(C\).
  Then \(I':=I\cap V_Y\)~is
  an independent set for~\(G_Y\).
  For \(S':=I'\cap Y\),
  we have that
  \(S'\cap X\cap Y=I'\cap X\cap Y=I\cap X\cap Y=S\cap X\cap Y\).
  Since \(I\setminus I'\subseteq X\)
  and,
  by \cref{def:td}\eqref{td3},
  \(I'\cap X\subseteq Y\),
  one has
  \(I\setminus I'=(I\cap X)\setminus (I'\cap X)=(I\cap X)\setminus (I'\cap X\cap Y)=(I\cap X)\setminus (I'\cap Y)=S\setminus S'\). 
  Thus,
  the vertices of~\(I'\)
  have pairwise distinct colors
  from~\(C':=C\setminus \col(S\setminus S')\)
  and
  \begin{align*}
    T[X,C,S]&=w(I)=w(I')+w(S\setminus S')\\
    &\leq T[Y,C\setminus\col(S\setminus S'),S']+w(S\setminus S').
  \end{align*}

  ``\(\geq\)'':
  Consider any independent set~\(S'\subseteq Y\)
  with \(S\cap X\cap Y=S'\cap X\cap Y\).
  Let \(I'\)~be
  a maximum\hyp weight independent set
  for~\(G_Y\)
  with
  \(I'\cap Y=S'\)
  whose vertices use pairwise distinct colors
  from~\(C':=C\setminus \col(S\setminus S')\).

  Then, \(I:=I'\cup S=I'\cup (S\setminus S')\)
  uses pairwise distinct colors in~\(C\).
  Moreover,
  by \cref{def:td}\eqref{td3},
  \(I'\cap X\subseteq V_Y\cap X\subseteq Y\)
  and, therefore, %
  \(I\cap X=(I'\cup S)\cap X=(I'\cap X)\cup(S\cap X)=(I'\cap X\cap Y)\cup S=(S'\cap X\cap Y)\cup S=(S\cap X\cap Y)\cup S=S\).
  
  We show that \(I\)~is independent.
  Towards a contradiction,
  assume that
  there is an edge~\(e=\{x,y\}\subseteq I=I'\cup S\).
  Since \(e\nsubseteq I'\) and \(e\nsubseteq S\),
  we have
  \(x\in S\setminus I'\subseteq X\) and
  \(y\in I'\setminus S\subseteq V_Y\).
  Since
  \(x\notin I'\supseteq  S'\cap X\cap Y=S\cap X\cap Y\)
  but \(x\in S=I\cap X\), we get \(x\notin Y\).
  Consequently,
  \(x\notin V_Y\) by \cref{def:td}\eqref{td3}
  and \(\{x,y\}\subseteq X\)
  by \cref{def:td}\eqref{td2}.
  Thus, we showed that
  \(\{x,y\}\subseteq X\cap I=S\),
  a contradiction.

  Finally
  since \(I'\cap X\subseteq I'\cap Y = S'\),
  sets~\(I'\) and \(S\setminus S'\) are disjoint
  and  we get
  \[
    T[X,C,S]\geq w(I)=w(I')+w(S\setminus S')=T[Y,C',S']+w(S\setminus S').
  \]
  
  \paragraph{Bags with two children.}
  For each bag~\(X\in\mathcal X\)
  with two children~\(Y=Z=X\),
  we show that
  \begin{gather}
    \label{eq:twochildren}
    \begin{split}
      T&[X,C,S]=\\
      &\max
      \begin{Bmatrix*}[l]
        T[Y,C_Y,S]+T[Z,C_Z,S]-w(S)\\
        \text{\quad such that } C_Y\cup C_Z=C\\
        \text{\qquad and }C_Y\cap C_Z=\col(S)
      \end{Bmatrix*}
      .
    \end{split}
  \end{gather}
  
  ``\(\leq\)'': Let \(I\)~be a maximum\hyp weight
  independent set for~\(G_X\) with \(X\cap I=S\)
  whose vertices use
  pairwise distinct colors
  from~\(C\).
  Then,
    \(I_Y:=I\cap V_Y\) and \(I_Z:=I\cap V_Z\)
  are independent sets for~\(G_Y\)
  and~\(G_Z\),
  respectively.
  One has
  \(Y\cap I_Y=Y\cap I\cap V_Y=Y\cap I=X\cap I=S\)
  and, likewise,
  \(Z\cap I_Z=S\).

  By \cref{def:td}\eqref{td3},
  one has \(I_Y\cap I_Z\subseteq X\)
  and thus \(I_Y\cap I_Z=I_Y\cap I_Z\cap X=S\).
  Since \(I\)~contains only one vertex
  of each color from~\(C\),
  we also get \(\col(I_Y)\cap\col(I_Z)=\col(S)\).
  Thus,
  the vertices of~\(I_Y\)
  have pairwise distinct colors
  in~\(C_Y:=\col(I_Y)\)
  and
  the vertices of~\(I_Z\)
  have pairwise distinct
  colors in~\(C_Z:=(C\setminus C_Y)\cup\col(S)\).
  One has \(C_Y\cup C_Z=C\) and \(C_Y\cap C_Z=\col(S)\).
  Thus,
  \begin{align*}
    T[X,C,S]=w(I)&=w(I_Y)+w(I_Z)-w(I_Y\cap I_Z)\\
                 &=w(I_Y)+w(I_Z)-w(S)\\
                 &\leq T[Y,C_Y,S]+T[Z,C_Z,S]-w(S).
  \end{align*}

  ``\(\geq\)'':
  Let \(C_Y\)~and \(C_Z\)~be
  color sets such that
  \(C_Y\cap C_Z=\col(S)\)
  and \(C_Y\cup C_Z=C\).
  Let \(I_Y\) and \(I_Z\)~be
  maximum\hyp weight independent sets
  for~\(G_Y\) and~\(G_Z\)
  with \(I_Y\cap Y=I_Z\cap Z=S\)
  and whose vertices use pairwise distinct
  colors from~\(C_Y\) and~\(C_Z\), respectively.
  Then, for
  \(I:=I_Y\cup I_Z\),
  one has
  \(I\cap X=(I_Y\cup I_Z)\cap X=(I_Y\cap X)\cup(I_Z\cap X)=S\).

  We show that \(I\)~is
  an independent set.
  Assume towards a contradiction that there is an
  edge~\(e=\{y,z\}\subseteq I=I_Y\cup I_Z\).
  Since \(e\nsubseteq I_Y\) and \(e\nsubseteq I_Z\),
  we have \(y\in I_Y\) and \(z\in I_Z\).
  \cref{def:td}\eqref{td3}
  gives \(V_Y\cap V_Z\subseteq X\).
  Thus,
  by \cref{def:td}\eqref{td2},
  \(e\subseteq V_Z\) or~\(e\subseteq V_Y\).
  By symmetry,
  assume that \(e\in V_Z\).
  Then, \(y\in V_Y\cap V_Z\subseteq X\).
  Since \(y\in X\cap I_Y=X\cap I_Z\),
  it follows that \(\{y,z\}\subseteq I_Z\),
  a contradiction to \(I_Z\)~being an independent set.

  We show that the
  colors of the vertices in~\(I\)
  are pairwise distinct.
  Let there be two vertices~\(\{u,v\}\subseteq I\)
  with \(\col(u)=\col(v)=c^*\).
  Then \(c^*\in\col(S)\)
  since \(C_Y\cap C_Z=\col(S)\).
  Thus,
  there is a vertex \(w\in S\subseteq I_Y\cap I_Z\)
  with \(\col(w)=c^*\).
  Since \(I_Y\) and \(I_Z\)
  contain only one vertex of each color, it holds that~$w = u = v$.
  Hence~$I$ only contains one vertex with color~$c^*$.
  Finally,
  \cref{def:td}\eqref{td3}
  yields \(I_Y\cap I_Z\subseteq X\).
  Thus,
  \begin{align*}
    T[X,C,S]\geq w(I)&=w(I_Y)+w(I_Z)-w(I_Y\cap I_Z)\\
                     &=w(I_Y)+w(I_Z)-w(I_Y\cap I_Z\cap X)\\
                     &=w(I_Y)+w(I_Z)-w(S)\\
                     &=T[Y,C_Y,S]+T[Z,C_Z,S]-w(S).
  \end{align*}

  \paragraph{The algorithm and its running time.}
  There are \(O(n)\)~bags.
  The algorithm first computes
  \eqref{eq:nochildren} for all leaf bags.
  Then,
  for each bag
  whose children have been processed,
  it computes
  \eqref{eq:onechild} and \eqref{eq:twochildren}.

  For each leaf bag~\(X\),
  we compute \eqref{eq:nochildren} as follows.
  First, for each
  of the at most \(n^\alpha\) subsets
  \(S\subseteq X\) with \(|S|\leq\alpha\),
  we once compute \(w(S)\)
  and check the independence of~\(S\)
  in \(O(\alpha^2)\)~time.
  We then use this information
  in the computation for all subsets
  \(C\subseteq\{1,2,\ldots,c\}\).
  Thus,
  all leaf nodes are processed in
  \(O(n\cdot (2^c+\alpha^2\cdot n^{\alpha}))\)~time.

  For each bag~\(X\in\mathcal X\)
  with two children,
  we compute \eqref{eq:twochildren} as follows.
  First, for each
  of the at most \(n^\alpha\) subsets
  \(S\subseteq X\) with \(|S|\leq\alpha\),
  we once compute \(w(S)\)
  and check the independence of~\(S\)
  in \(O(\alpha^2)\)~time.
  We then use this information
  in the computation for all subsets
  \(C\subseteq\{1,2,\ldots,c\}\).
  To this end,
  iterate over all subsets~\(C_Y\subseteq C\),
  whereby \(C_Z\)~can be computed from~\(C_Y\)
  in \(O(c)\)~time.
  Note that,
  in total,
  one iterates over at most \(3^c\)~subsets
  \(C_Y\subseteq C\subseteq \{1,2,\ldots,c\}\):
  each color is
  either not in~\(C\),
  or in \(C\) but not in~\(C_Y\),
  or in~\(C_Y\).
  Thus,
  bags with two children
  can be processed in total time
  \(O(n\cdot (3^cc+\alpha^2\cdot n^{\alpha}))\).
  
  For each bag~\(X\in\mathcal X\)
  with a single child~\(Y\),
  we compute \eqref{eq:onechild} as follows.
  First, for each
  of the at most \(n^\alpha\) subsets
  \(S'\subseteq Y\) with \(|S'|\leq\alpha\),
  we once compute the intersection \(S'\cap X\) and
  \(w(S'\cap X)\)
  in \(O(\alpha^2)\)~time.
  During this computation,
  for each encountered intersection~\(S^*\)
  and color set \(C\subseteq\{1,2,\ldots,c\}\),
  remember
  the maximum value \(T_{S^*}[Y,C]\)
  taken by \(T[Y,C,S']\) for any
  \(S'\subseteq Y\) with \(S^*=S'\cap X\).
  This works in \(O(2^c\cdot\alpha)\)~time.
  We then iterate
  over
  all subsets
  \(C\subseteq\{1,2,\ldots,c\}\),
  each of the encountered intersections~\(S^*\subseteq X\cap Y\),
  and each independent set~\(S\) with \(S^*\subseteq S\subseteq X\)
  and
  compute
  $T[Y,C\setminus\col(S\setminus S'),S']$
  as the maximum value
  $T_{S^*}[C\setminus\col(S\setminus S^*)]-w(S\setminus S^*)$
  encountered for any intersection~\(S^*\) that yielded~\(S\).
  For each~$\alpha' \leq \alpha$ and each of \(n^{\alpha'}\)~intersections~\(S^*\)
  of size exactly~\(\alpha'\),
  we enumerate at most
  \(n^{\alpha-\alpha'}\)~independent sets~\(S\)
  such that \(S^*\subseteq S\subseteq X\).
  Thus,
  bags with a single child can be processed in
  \(O(n\cdot (\alpha^2+2^c\alpha +2^c n^{\alpha}))\)~total time.
  
  Altogether, the overall algorithm runs in
  \(O(3^cn^{\alpha+1})\)~time.
  \qed
\end{proof}

\section{Parameterized complexity on other graph classes}
\label{sec:related}

In \cref{sec:ishard},
we have shown that
\IS{} is W[1]-hard parameterized by the solution size
even on 2\hyp simplicial 3-minoes.
In contrast,
in \cref{sec:iseasy},
we have shown that
\McCS{} is \fp{} tractable parameterized
by the solution size on
cluster\({}\bowtie{}\)chordal graphs.

In this section,
we survey the parameterized complexity
of \McCS{}
on the neighboring graph classes
discussed in \cref{sec:classes},
thus completing
the computational complexity picture
given in \cref{fig:MIS}.
Most results in this section are known
or easy to obtain.
Yet since they are
scattered throughout the literature
using terminology from
quite diverse subject fields,
we summarize them here.

\paragraph{$K_{1,3}$-free graphs.}
\cref{sec:ishard}
has shown that \IS{} is W[1]-hard parameterized
by the solution size
on 3-minoes,
which are a proper subclass of $K_{1,4}$-free graphs.
This is complemented as follows.

\begin{proposition}
  \label[proposition]{sec:k13free}
  \label[proposition]{prop:cfhard}
  On \(K_{1,3}\)-free graphs,
  \begin{enumerate}[i)]
  \item\label{k13hard} \cCS{} is NP-hard for each fixed~$c \geq 3$ and
  \item\label{k13easy} a maximum\hyp weight \(c\)\hyp colorable subgraph
  with at most \(\ell\)~vertices
  is computable in \(2^{\ell+c}\cdot c^\ell\cdot \ell^{O(\log \ell)}\cdot n^3\log^2n\)~time.
\end{enumerate}
\end{proposition}

\begin{proof}
  \cref{prop:cfhard}\eqref{k13hard}
  follows from the fact that
  checking 3\hyp colorability
  of line graphs
  is NP\hyp complete,
  where line graphs are $K_{1,3}$-free.

\cref{prop:cfhard}\eqref{k13easy}
follows directly from \cref{lem:mccs-to-is}
using the fact
that \(K_{1,3}\)-free graphs
are a hereditary graph class
and that
\MIS{} is solvable in \(O(n^3)\)~time 
on \(K_{1,3}\)-free graphs \citep{FOS11}.
  \qed
\end{proof}

\noindent
Note that \(c\leq\ell\) holds in all nontrivial cases.
Thus,
\cref{prop:cfhard}
shows that
\McCS{} is \fp{} tractable
parameterized by~\(\ell\)
on \(K_{1,3}\)-free graphs.
The complexity
for the case \(c=2\)
seems to be open.

\paragraph{Cluster\({}\bowtie{}\)interval graphs.}
These graphs form a subclass of cluster\({}\bowtie{}\)chordal
graphs
and therefore \cref{thm:iseasy} applies to them.
Yet for cluster\({}\bowtie{}\)interval graphs,
one can give a slightly better running time.

\begin{proposition}
  \label[proposition]{sec:cluster-interval}
  On cluster\({}\bowtie{}\)interval graphs,
  \begin{enumerate}[i)]
  \item\label{ci1} \MIS{} is NP-hard and
  \item\label{ci2} a maximum\hyp weight \(c\)\hyp colorable subgraph
  with at most \(\ell\)~vertices
  is computable in \(2^{\ell+c}\cdot (c+e+2)^{\ell}\cdot \ell^{O(\log\ell)}\cdot (n+m)\cdot \log^3 n\)~time
if the decomposition of the input graph
into a cluster graph and an interval graph is given.
\end{enumerate}
Herein, \(e\)~is Euler's number.
\end{proposition}
\noindent
\cref{sec:cluster-interval}\eqref{ci1}
is due to \citet{HK06}.

Towards proving
\cref{sec:cluster-interval}\eqref{ci2},
\citet{BMNW15} noted that
\citet{HK06} showed that
a maximum\hyp weight independent set
of vertices of pairwise distinct colors
in an interval graph whose vertices
are colored in \(c\)~colors
is computable in \(O(2^c\cdot n)\)~time
if an interval representation
with sorted intervals is given,
which can be computed in \(O(n+m)\)~time.
Since cluster\({}\bowtie{}\)interval graphs
are a hereditary graph class,
\cref{sec:cluster-interval}\eqref{ci2}
follows from
this result of \citet{HK06} and
\cref{lem:mccs-to-is,lem:clu-to-co}.

\paragraph{Chordal graphs.}
\cref{thm:ishard} shows that \IS{}
is W[1]-hard parameterized by the solution size
on inductive \ind{2} graphs.
Chordal graphs are
exactly the inductive \ind{1} graphs
\citep{BP93,YB12}.
\begin{proposition}\label[proposition]{sec:chordal}
  In chordal graphs, \McCS{} is
  \begin{enumerate}[i)]
  \item\label{chord1} NP-hard,
  \item\label{chord3} polynomial\hyp time solvable for each fixed~\(c\) yet W[2]-hard parameterized by~\(c\), and

  \item\label{chord2} a maximum\hyp weight \(c\)\hyp colorable subgraph
  with at most \(\ell\)~vertices
  is computable in \(2^{\ell+c}\cdot c^{\ell}\cdot \ell^{O(\log\ell)}\cdot (n+m)\cdot \log^3 n\)~time.
\end{enumerate}
\end{proposition}
\noindent
\cref{sec:chordal}\eqref{chord1} and \eqref{chord3}
were shown by \citet{YG87}.
\cref{sec:chordal}\eqref{chord2}
follows from the fact that
\MIS{} is linear\hyp time solvable in chordal graphs \citep{Fra75}
and using \cref{lem:mccs-to-is}.
This was shown for the unweighted variant
by \citet{MPR+13}
and our \cref{lem:mccs-to-is} is,
essentially, just a slightly adapted weighted variant
of this result.

\paragraph{Interval graphs.}
\citet{YG87} noted that
\McCS{} in interval graphs
is solvable
in polynomial time
by modelling it as a
totally unimodular
linear program.

Using a flow formulation,
\citet{AS87} proved that the
equivalent problem
of scheduling a maximum\hyp weight subset
of intervals
on \(c\)~parallel identical machines
such that the intervals on each machine
are pairwise nonintersecting
is solvable in \(O(n^2\log n)\)~time.
Since the interval representation
of an interval graph is computable
in the same time \citep{COS09},
one gets the following result.

\begin{proposition}
  \label[proposition]{sec:interval}
  \McCS{} is solvable in \(O(n^2\log n)\)~time on interval graphs.
\end{proposition}

\paragraph{Inductive \ind{k} graphs.}
Of course,
our hardness result of \cref{thm:ishard}
also holds for \IS{} in inductive \ind{k}
graphs.

On the positive side,
\citet{YB12} showed that \MIS{}
is polynomial-time \(k\)\hyp approximable
in inductive \ind{k} graphs.
The algorithm requires the
\emph{\indc{k} ordering}
of \cref{def:iki},
which cannot be computed efficiently
(cf. \cref{prop:indkind-hard}),
yet sometimes is given
by the application data \citep{AHT17,HT15}.

\section{Conclusion}
\label{sec:concl}
Motivated by recent work on wireless scheduling \citep{AHT17,HT15}, 
we performed 
an extensive analysis of problems related to independent sets
on graphs that lie between interval
graphs on the ``bottom end'' and inductive \ind{2} graphs on the 
``top end''. Some of our (computational hardness) 
results might be discouraging given
that,
in real-world scheduling applications,
one often has to expect
inductive $k$\hyp independent graphs with constant~$k>2$. 
This negative impression is alleviated by the following observations.

First, there are several tractability results for classic scheduling 
problems
that can be modeled as independent set problems at this 
fairly low level in the graph classes hierarchy.

Second, 
our parameterized complexity studies are basically focused on the 
parameter solution size. In the spirit of multivariate complexity 
analysis \citep{FJR13,KN12,Nie10} this encourages the studies 
of further and also ``combined'' parameterizations. Indeed, 
with real-world data at hand, one might first measure numerous parameter
values in the instance (say characteristics such as maximum degree, treewidth, 
degree distribution, feedback edge and vertex numbers) that, if small,
altogether might be exploited to get strong fixed\hyp parameter 
tractability results. %

Third,
we point out that combining approximation algorithms
(there are several results in the 
literature including \cite{AHT17,HT15}; and \citet{YB12})
with problem\hyp specific parameterizations
might lead to parameterized approximation algorithms
with improved approximation factors and/or running times
in relevant application scenarios, see \citet{BKS17} for a concrete example. 

We remark that our main positive result 
(fixed\hyp parameter tractability of finding maximum\hyp weight $c$\hyp colorable subgraphs
on cluster\({}\bowtie{}\)chordal graphs, \cref{thm:iseasy}) indeed also holds for the so far
seemingly neglected class of graphs with a given tree decomposition where
the independent sets in each bag have size at most~$\alpha$: it can be solved in 
$O(3^c\cdot n^{\alpha +1})$ time (\cref{lem:td}). 
We feel that this class of graphs might be 
of independent interest.

Finally, \cref{sec:classes} with its classification 
of graph classes between interval graphs and inductive \ind{3} graphs
might prove useful even for (graph\hyp algorithmic) studies not touching 
scheduling and 
independent set problems. On a different note, we advocate
that
parameterized complexity analysis may prove useful for further studies on
hard scheduling problems---so far the number of studies in this 
direction is fairly small and leaves numerous challenges for future 
work \citep{MBxx}.

\paragraph{Acknowledgement.}
We are grateful to Andreas Krebs (Tübingen) for fruitful discussions concerning parts of this work.
We thank the anonymous referees
of \emph{Journal of Scheduling}
for constructive feedback.

\bibliographystyle{spbasic}
\bibliography{iki}

\end{document}